%% file: ms.tex
\documentclass[acmsmall,screen,nonacm]{acmart}
\acmConference[PL'18]{ACM SIGPLAN Conference on Programming Languages}{January 01--03, 2018}{New York, NY, USA}
\acmYear{2018}
\acmISBN{} 
\acmDOI{} 
\startPage{1}

\setcopyright{none}

\usepackage{booktabs}   
\usepackage{subcaption} 

\usepackage{bussproofs}
\usepackage{mdframed}
\usepackage{tikz-cd}
\usepackage{mathtools}
\usepackage[frozencache]{minted}

\DeclareMathOperator*{\trans}{\mathsf{trans}}
\DeclareMathOperator{\phitolambda}{\trans_{\varphi\to\lambda}}

\DeclareMathOperator{\inc}{\mathsf{inc}}
\DeclareMathOperator{\redto}{\rightsquigarrow}
\DeclareMathOperator{\redtomany}{\stackrel{\ast}{\rightsquigarrow}}
\DeclareMathOperator{\parredtomany}{\stackrel{\ast}{\Rrightarrow}}
\DeclareMathOperator{\headred}{\stackrel{\textit{h}}{\rightsquigarrow}}
\DeclareMathOperator{\leftmost}{\stackrel{\textit{n.o.}}{\rightsquigarrow}}
\DeclareMathOperator{\innerred}{\stackrel{\textit{i}}{\rightsquigarrow}}
\DeclareMathOperator{\headredmany}{\stackrel{\textit{h}^\ast}{\rightsquigarrow}}
\DeclareMathOperator{\innerredmany}{\stackrel{\textit{i}^\ast}{\rightsquigarrow}}
\DeclareMathOperator{\leftmostmany}{\stackrel{\textit{n.o.}^\ast}{\rightsquigarrow}}
\DeclareMathOperator{\innerparred}{\stackrel{\textit{i}}{\Rrightarrow}}

\newcommand{\reccat}[2]{#1\;\|\;#2}
\newcommand{\recext}[2]{#1\;\text{\textbf{with}}\;#2}

\newcommand{\mkObject}[1]{\llbracket #1 \rrbracket}
\newcommand{\inct}[2]{#2\uparrow{}^{#1}}
\newcommand{\mkState}[1]{\langle #1 \rangle}

\DeclareMathOperator{\attr}{\mathsf{attr}}
\DeclareMathOperator{\fix}{\mathsf{fix}}

\begin{document}

\title[Formalizing $\varphi$-calculus]{Formalizing $\varphi$-calculus: a purely object-oriented calculus of decorated objects}         


\author{Nikolai Kudasov}
\orcid{0000-0001-6572-7292}             
\affiliation{
  \institution{Innopolis University}            
  \streetaddress{Universitetskaya 1}
  \city{Innopolis}
  \state{Tatarstan Republic}
  \postcode{420500}
  \country{Russia}
}
\email{n.kudasov@innopolis.ru}          

\author{Violetta Sim}
\affiliation{
  \institution{Innopolis University}            
  \streetaddress{Universitetskaya 1}
  \city{Innopolis}
  \state{Tatarstan Republic}
  \postcode{420500}
  \country{Russia}
}
\email{v.sim@innopolis.university}         

\input{sections/0-abstract}



\maketitle

\input{sections/1-introduction}

\input{sections/2-calculus}
\input{sections/3-confluence}
\input{sections/4-abstract-machine}
\input{sections/5-translation}
\input{sections/6-extensions}
\input{sections/7-related-work}
\input{sections/8-conclusion}

\begin{acks}                            
  This research has been generously funded by Huawei in the framework of Polystat project.
  We thank Yegor Bugayenko for taking his time to explain ideas behind EO, his version of $\varphi$-calculus, and especially his vision regarding decorator and parent object locators. We thank Bertrand Meyer for giving his feedback on the early version of the calculus and suggesting terminology for ``void'' and ``attached'' attributes. 
  We also thank Nickolay Shilov and Larisa Safina for their feedback on the paper and different versions of the calculus, and Georgii Gelvanovskii, for proofreading the paper.
  Finally, we thank Luigi Liquori and anonymous reviewers of PLDI 2022, FTfJP 2022 and TOPLAS for thorough reading and helpful suggestions on earlier versions of this paper.
\end{acks}

\bibliography{sample-base}

\clearpage
\appendix

\input{sections/A-complete-proofs.tex}

\end{document}

%% file: sections/0-abstract.tex
\begin{abstract}
  Many calculi exist for modelling various features of object-oriented languages. Many of them are based on $\lambda$-calculus and focus either on statically typed class-based languages or dynamic prototype-based languages.
  We formalize untyped calculus of decorated objects, informally presented by \citeauthor{bugayenko/online}, which is defined in terms of \emph{objects} and relies on \emph{decoration} as a primary mechanism of object extension. It is not based on $\lambda$-calculus, yet with only four basic syntactic constructions is just as complete. We prove the calculus is confluent (i.e. possesses Church-Rosser property), and introduce an abstract machine for call-by-name evaluation. Finally, we provide a sound translation to $\lambda$-calculus with records.
\end{abstract}

%% file: sections/1-introduction.tex
\section{Introduction}



Recently, Bugayenko \cite{bugayenko/online} introduced a new programming language EO and its semantics in terms of an informally specified calculus which he calls $\varphi$-calculus. The EO language incorporates Decorator pattern \cite[Chapter 4]{GoF1995} as the only mechanism of extending objects. This is somewhat similar to delegation-based inheritance, which makes EO close in spirit to Self~\cite{UngarSmith1987}. Bugayenko's paper lays out interesting ideas, but suffers from inaccuracies and insufficient formalization of the calculus, as the paper lacks reduction semantics and any soundness results. In this paper, we formalize the key ideas of Bugayenko's $\varphi$-calculus.

\subsection{Existing formalisations}

Formalizing object-oriented features of programming languages is an old but still vibrant topic in computer science. Many of the formal models of object-oriented languages, extensions of which are used and studied today, have emerged in the 1990s.

Early formalisations intended for immediate use in software development go at least to VDM++~\cite{durr_vdm_1992} and Object-Z~\cite{DUKE1995511} which both allow some formal reasoning about classes, objects, and inheritance. Object-Z also has support for multiple inheritance and polymorphism (method overloading).

At the same time, type theoretic models for object-oriented features appear from variations on $\lambda$-calculus. Pierce and Turner's \emph{``Simple type-theoretic foundations for object-oriented languages''}~\cite{pierce_simple_1993} and Cardelli's \emph{pure calculus of subtyping}~\cite{Cardelli1994ExtensibleRI} develop a formal type-theoretic account for the basic features of object-oriented programming in a purely functional setting of typed $\lambda$-calculus.

Abadi and Cardelli's \emph{theory of primitive objects}~\cite{ABADI199678} offers $\varsigma$-calculus, which is a $\lambda$-calculus-like formalism that relies on an external memory to store and modify states of \emph{mutable} objects. The calculus has a sound type system with basic subtyping via subsumption. This work has been expanded into $\text{\textbf{FOb}}_{<:\mu}$ calculus, a variation of System F~\cite[Chapter 11]{Girard1989}.

Featherweight Java (FJ)~\cite{IgarashiPierceWadler2001} introduces a minimal core calculus for Java, using a \emph{nominal} type system. FJ focuses on representing the minimal core of Java, omitting many features, including mutable state. Welterweight Java~\cite{OstlundWrigstad2010} adds a few extra pounds to FJ making it imperative, stateful, thread-based concurrency and lock synchronization. Welterweight Java positions itself as a good starting point for extensions. A recent such extension, OOlong~\cite{CastegrenWrigstad2019} presents a concurrent object calculus aimed at extensibility and reuse. As such OOlong provides mechanised version for rigorous proofs in Coq.

Castagna, Ghelli and Longo introduced a \emph{calculus for overloaded functions with subtyping}~\cite{Castagna1992ACF} (based on $\lambda$-calculus), offerring a model able to represent object-oriented languages with multiple dispatch (mutli-methods). More recent work on formalisations of multiple dispatch includes variations of Featherweight Generic Fortress language~\cite{AllenLuchangcoRyuSteele2007, AllenHilburnKilpatrickLuchangcoRyuChaseSteele2011, ParkHongSteeleRyu2019}, Featherweight Hierarchical Java~\cite{WangZhangOliveiraServetto2018}, and \emph{prototypes with multiple dispatch}~\cite{SalzmanAldrich2005}.

Ababi's semantics of Baby Modula-3~\cite{Abadi1994} and Mitchell, Honsell and Fisher's \emph{Lambda Calculus of Objects, $\lambda\mathrm{Obj}$}~\cite{Fisher1993ALC} both extended $\lambda$-calculus with new syntactic forms to model delegation-based inheritance (a la Self~\cite{UngarSmith1987}). $\lambda\mathrm{Obj}$ supports destructive operations on objects: method addition and method override. The operational semantics of the calculus allows objects to modify themselves, which is a self-inflicted operation~\cite{Cardelli1995}. The most recent work in this direction is by Ciaffaglione, Di Gianantonio, Honsell and Liquori~\cite{Ciaffaglione2021}, who introduce $\lambda\mathrm{Obj}^\oplus$ system, an extension of $\lambda\mathrm{Obj}$ with additional type system features enabling reasoning about \emph{object evolution} and \emph{object reclassification}.

Most of the above models differ in their approaches to subtyping and subclasses, but rely significantly on $\lambda$-calculus as their foundation. At the same time, many modern programming languages do not properly support $\lambda$-abstraction. For example, Java provides lambda expressions which are essentially instances of anonymous classes with a single method. Also, when modelling object-oriented languages, many formalizations focus on statically typed class-based languages while a few prefer prototype-based approach \cite{Ciaffaglione21,SalzmanAldrich2005}.

\subsection{Contribution}

In this paper, we extract the key ideas from Bugayenko's paper and formalize $\varphi$-calculus, a calculus of objects with decoration as the main mechanism for object extension, in a more formal way. Moreover, our interpretation of $\varphi$-calculus is not based on $\lambda$-calculus, so object methods are also represented as objects.

In contrast with Bugayenko's work, we focus on the object-oriented core of the calculus, and do not introduce primitives such as numbers, booleans, mutable memory. Similarly to $\lambda$-calculus, these primitives can either be added via a straightforward extension, or by using an encoding, such as Church numerals. We leave details of such extensions and encodings outside the scope of this paper.

Our specific contributions are the following:


\begin{itemize}
  \item We present syntax and introduce reduction semantics of $\varphi$-calculus, a calculus of objects with decoration.
  \item We prove that $\varphi$-calculus possesses Church-Rosser property.
  \item We define normal order evaluation of $\varphi$-terms and prove its completeness, i.e. any term will be reduced to its normal form under such evaluation order, whenever the normal form exists.
  \item We define an abstract machine for call-by-name evaluation of $\varphi$-terms.
  \item We introduce a translation into $\lambda$-calculus with records that maps objects almost directly into records and prove this translation to be sound. We also show that $\varphi$-calculus is Turing-complete by sketching an embedding of pure $\lambda$-calculus (without records) into $\varphi$-calculus.
  \item We show how to extend $\varphi$-calculus with primitives and syntactic sugar to match capabilities of EO programming language.
\end{itemize}

We note that although we are trying to preserve the original notation and terminology of Bugayenko~\cite{bugayenko/online}, in his paper he uses custom terminology (such as \emph{free} and \emph{bound} attributes) that conflicts the established conventions in the computer science literature. In such cases, we use different terminology (such as \emph{void} and \emph{attached} attributes).

%% file: sections/2-calculus.tex
\section{Calculus of decorated objects}
\label{section:calculus}

In this section, we introduce syntax and evaluation rules, providing the intuition behind those. The central concept in $\varphi$-calculus is that of an object. In fact, every term in $\varphi$-calculus is essentially an object.

\begin{definition}
  \label{definition:object}
    Assuming a set of labels $\mathcal{L}$ and a set of terms $T$, an \emph{object} $X$ is a mapping from $\mathcal{L}$ to $T \cup \{\varnothing, \bot\}$. The labels that map to $\bot$ are called \emph{missing}, the labels that map to $\varnothing$ are called \emph{void attributes} of $X$ and labels that map to terms are called \emph{attached attributes} of $X$. Void and attached attributes of $X$ are collectively called \emph{attributes} of $X$ and are denoted $\attr(X) \subseteq \mathcal{L}$. \\
    An object with an non-empty set of void attributes is called \emph{abstract}. Otherwise an object is called \emph{concrete}.
\end{definition}

\begin{example}
  A constant mapping from $\mathcal{L}$ to $\bot$ is an \emph{empty object}
  and is denoted $\mkObject{}$.
\end{example}

\begin{example}
  \label{example:object}
  Let $t_1, t_2$ be terms. Then the following is an object:
  \begin{align*}
    X &: \mathcal{L} \to T \cup \{\varnothing, \bot\} \\
    x &\mapsto \varnothing; y \mapsto t_1; z \mapsto \mkObject{}; \ell \mapsto \bot
  \end{align*}
\end{example}

For convenience we will denote objects by listing void and attached attributes inside double brackets. For example, the object from Example~\ref{example:object} can be written succinctly as $\mkObject{x \mapsto \varnothing, y \mapsto t_1, z \mapsto \mkObject{}}$.

\subsection{Syntax}

The entire syntax of $\varphi$-calculus has only four syntactic constructions:

\begin{definition}
  Let $\mathcal{L}$ be the set of attribute names extended with \emph{decorator attribute} $\varphi$. Then the set of $\varphi$-terms $T$ is defined inductively as following:
  \begin{enumerate}
    \item if $n \in \mathbb{N}$ then $\rho^{n} \in T$; here $\rho$ is merely a symbol used in the syntax, it is not a variable or a meta variable;
    \item if $t \in T$ and $a \in \mathcal{L}$ then $t.a \in T$;
    \item if $t, u \in T$ and $a \in \mathcal{L}$ then $t(a \mapsto u) \in T$;
    \item if $t_1, \ldots, t_n \in T$ and $a_1, \ldots, a_k, b_1, \ldots, b_n \in \mathcal{L}$ then $\mkObject{a_1 \mapsto \varnothing, \ldots, a_k \mapsto \varnothing, b_1 \mapsto t_1, \ldots b_n \mapsto t_n} \in T$.
  \end{enumerate}

  The instance of rule~1 is called \emph{locator}. The instance of rule~2 is called \emph{attribute access}. The instance of rule~3 is called \emph{application} or \emph{attribute instantiation}. The instance of rule~4 is called \emph{object term} (or just \emph{object}).
\end{definition}
An object term is essentially the same as object from Definition~\ref{definition:object} with the restriction that mapping has finitely many attributes.

Locators allow to reference enclosing objects, and consequently, their attributes.
For example, consider object $\mkObject{x \mapsto \rho^0.y, y \mapsto \mkObject{}}$. Here $\rho^0$ references the closest enclosing object, which is the entire term, and so attribute $x$ is attached to a term that essentially references attribute $y$ of the same object. The index in the locator tells us how many levels of enclosing objects one needs to go to arrive at the referenced object. In fact, locators are effectively de Bruijn indices \cite{deBruijn1972} of nested objects. For example, locator $\rho^2$ in the term $\mkObject{x \mapsto \mkObject{y \mapsto \mkObject{z \mapsto \rho^2}}}$ references the entire term. This correspondence is made more precise in Section~\ref{section:translation-to-lambda-calculus} where we give a translation to lambda calculus.

The idea behind attribute access terms is fairly straightforward: we simply intend to extract the associated value. For example, evaluating $\mkObject{x \mapsto \mkObject{}}.x$ should produce an empty object. However, because of locators, we cannot do a simple extraction in general and have to perform locator substitution: evaluating $\mkObject{x \mapsto \rho^0}.x$ requires understanding what object $\rho^0$ references. This limitation makes it very different from syntactically similar construction in $\lambda$-calculus with records.

Attribute instantiation attaches terms to void attributes. Importantly, attached attributes may reference void attributes as in $\mkObject{x \mapsto \varnothing, y \mapsto \rho^0.x}$. Obviously, accessing attribute $y$ of such an object would require accessing void attribute $x$ which does not have an attached value. However, after instantiating the attribute $\mkObject{x \mapsto \varnothing, y \mapsto \rho^0.x}(x \mapsto \mkObject{})$, we can now access attribute $y$ and get the empty object. This example shows that locators are not just references to syntactic objects, but can also reference the result of attribute instantiation.

Decorator attribute $\varphi$ plays an important role in evaluation. This attribute contains the component of the decorator, and an object with $\varphi$ will redirect attribute access to its component whenever the object does not possess the attribute itself. For example, accessing attribute $.x$ of the object $\mkObject{\varphi \mapsto \mkObject{x \mapsto t_1}, y \mapsto t_2}$ will be redirected to accessing $.\varphi.x$ of the object since the original object does not have $x$ as its attribute.

\subsection{Locators}

Before we can properly talk about evaluation, we need to define how locators work. As locators are de Bruijn indices of nested objects, we have to be able to adjust locator indices when moving terms in and out of objects, and replace locators with an actual object they reference when that object disappears.

Attribute instantiation requires putting a term in an object. This requirement may demand updating certain locators. Consider term $\mkObject{x \mapsto \varnothing}(x \mapsto \rho^0)$ where $\rho^0$ references some outer object. Simply replacing $\varnothing$ with $\rho^0$ will result in $\mkObject{x \mapsto \rho^0}$ which would change the meaning of $\rho^0$. Instead, since we are placing $\rho^0$ inside of an object, we have to increment its index. In general though, given $t(a \mapsto u)$ we do not increment all locators in $u$, but only those referencing outside of $u$. For example, if $u \equiv \mkObject{y \mapsto \rho^0, z \mapsto \rho^1}$ then we only increment $\rho^1$, since otherwise its reference object will change when we put $u$ in an object. So we define locator increment as follows:

\begin{definition}
    \emph{Locator increment} $\inct{n}{t}$ is defined inductively on $\varphi$-terms:
    \begin{align}
        \inct{n}{\rho^m} &:= \rho^{m} \quad \text{if $m < n$} \\
        \inct{n}{\rho^m} &:= \rho^{m+1} \quad \text{if $n \leq m$} \\
        \inct{n}{t.a} &:= \inct{n}{t}.a \\
        \inct{n}{t_1(a \mapsto t_2)} &:= \inct{n}{t_1}(a \mapsto \inct{n}{t_2}) \\
        \inct{n}{\mkObject{a_1 \mapsto \varnothing, \ldots, b_1 \mapsto t_1, \ldots, b_n \mapsto t_n}} \span \notag \\
        := \mkObject{a_1 \mapsto \varnothing, \ldots, b_1 \mapsto \inct{n+1}{t_1}, \ldots, b_n \mapsto \inct{n+1}{t_n}} \span
    \end{align}

    We write $\inct{}{t}$ short for $\inct{0}{t}$.
\end{definition}

When accessing an attribute $a$ of an object $u \equiv \mkObject{\ldots, a \mapsto t, \ldots}$, we cannot simply return $t$ as it may contain locators that need adjustment. We can split locators in $t$ into three groups:
\begin{enumerate}
  \item Locators that reference some object \emph{inside} of $t$ require no adjustment as corresponding objects are still present after extraction.
  \item Locators that reference object $u$, which we are extracting from, have to be substituted with $u$ itself, to avoid losing information when we extract the subterm $t$.
  \item Locators that reference some object \emph{outside} of $u$ have to be decremented since one nested object ($u$) disappeared.
\end{enumerate}

We define locator substitution with the usual notation $t[\rho^n \mapsto u]$ meaning that we intend to substitute all locators referencing the same object as $\rho^n$ in $t$ with term $u$ (with all locators updated properly).

\begin{definition}
  \emph{Locator substitution} $t[\rho^n \mapsto u]$ is defined inductively on $\varphi$-terms (see Figure~\ref{fig:core-phi-calculus}).
\end{definition}

\begin{definition}
  A $\varphi$-term $t$ is called \emph{closed} if all its locators reference some object in $t$. Otherwise it is \emph{open}.
\end{definition}

\subsection{Evaluation}

In this section we define small step reduction semantics for $\varphi$-calculus. We introduce four congruence rules, to enable reduction in subterms, as well as three main reduction rules. Attribute access and attribute instantiation provide two reduction rules given an object term. Then, considering the presence of the decorator attribute $\varphi$, we get one extra rule for attribute access. Figure~\ref{fig:core-phi-calculus} shows the complete set of reduction rules.

Rule DOT$_c$ formalizes the idea that extracting attribute $c$ from an object $t$ is straightforward as long as we substitute all locators that reference to $t$ in the resulting term. Rule DOT$_c^\varphi$ specifies further that whenever $c$ is not an attribute of $t$ but $\varphi$ is, we should extract $c$ from $t.\varphi$. Importantly, we do extract the term attached to $\varphi$ immediately. Such extraction would require to check recursively whether we should go to $t.\varphi.\varphi$ and further. Instead we want our rules to perform just a single step of reduction.

Rule APP$_c$ shows that attaching a term $u$ to the attribute $c$ requires incrementing locators in $u$.

\subsubsection{Decorated instantiation}

As we are following the idea of Decorator pattern \cite[Chapter 4]{GoF1995}, in $\varphi$-calculus only concrete objects are supposed to be decorated. That said, one could consider a variation of our calculus where decoration and later instantiation of abstract objects is possible as well. For example, we can introduce the following evaluation rule:

\begin{prooftree}
  \AxiomC{$t \equiv \mkObject{\ldots, \textcolor{black}{\varphi \mapsto t_\varphi}, \ldots}$}
  \AxiomC{$c \notin \attr(t)$}
  \RightLabel{APP$_{c}^\varphi$}
  \BinaryInfC{$t (c \mapsto u) \rightsquigarrow \mkObject{\ldots, \textcolor{black}{\varphi \mapsto t_\varphi (c \mapsto \inct{}{u})}, \ldots}$}
\end{prooftree}

For the rest of this paper we will assume $\varphi$-calculus without APP$_c^\varphi$ rule. However, all results still hold for a variation of the calculus with it.

\input{figures/core-calculus}

\subsection{Modelling object-oriented languages}

Bugayenko \cite{bugayenko/online} introduced $\varphi$-calculus as a semi-formal mathematical model for EO programming language. EO language is an object-oriented programming language that relies on decoration instead of inheritance to work with object hierarchies. In \cite{bugayenko2021reducing}, \citeauthor{bugayenko2021reducing} showed approaches to encode advanced features of mainstream object-oriented languages in EO.
In this subsection, we revisit \citeauthor{bugayenko2021reducing}'s ideas \cite{bugayenko2021reducing} regarding encoding mechanisms of code reuse in class-based and prototype-based constructs in terms of $\varphi$-calculus.

\subsubsection{Class-based}

Classes can be modelled as the so called ``object factories''~--- objects with a special method $\mathsf{new}$, that produces an instance of the class. To model inheritance, this method can take an object as input and extend it with all the necessary methods using decoration.

\begin{example}
  \label{example:phi-classes}
  Consider the following code snippet in Java:
\begin{minted}{java}
class Base { 
    Integer g() { return h(); }
    Integer h() { return 3; }
}
class Derived extends Base {
  Integer f() { return 2 + this.g(); }
  Integer h() { return 5; }
}
Derived d = new Derived();
\end{minted}

In the following translation scheme from Java to $\varphi$-calculus, it is made explicit that methods accept $\mathsf{this}$ argument. It allows to distinguish derived class, when it calls a method that is defined in a subclass.

\begin{align*}
    \mathsf{Base} &:= \mkObject{\mathsf{new} \mapsto  \\
                  &\phantom{:= \llbracket \mathsf{new}} \mkObject{ \mathsf{g} \mapsto \mkObject{\mathsf{this} \mapsto \varnothing, \;\varphi \mapsto \mathsf{\rho^0.this.h(this \mapsto \rho^0.this).\varphi}},\\
                  &\phantom{:= \llbracket \mathsf{new} \llbracket} \mathsf{h} \mapsto \mkObject{\mathsf{this} \mapsto \varnothing, \;\varphi \mapsto 3}}\\
                  &\phantom{:= \llbracket} }\\
 \mathsf{Derived} &:= \mkObject{\mathsf{new} \mapsto  \\
                  &\phantom{:= \llbracket \mathsf{new}} \mkObject{ \varphi \mapsto \mathsf{Base.new},\\
                  &\phantom{:= \llbracket \mathsf{new} \llbracket} \mathsf{f} \mapsto \mkObject{\mathsf{this} \mapsto \varnothing, \;\varphi \mapsto \mathsf{2.add(n\mapsto\rho^0.this.g(this \mapsto \rho^0.this).\varphi)}},\\
                  &\phantom{:= \llbracket \mathsf{new} \llbracket} \mathsf{h} \mapsto \mkObject{\mathsf{this} \mapsto \varnothing, \;\varphi \mapsto 5}}\\
                  &\phantom{:= \llbracket} } \\
d &:= \mathsf{Derived}.\mathsf{new}\span
\end{align*}

In Java, the method invocation \texttt{d.f()} computes to 7: the method \texttt{f()} of the derived class calls \texttt{g()}, defined in the base class, which, in turn, calls \texttt{h()}, which is overridden in the the derived class.
In $\varphi$-terms, with explicit $\mathsf{this}$ arguments, the chain of calls is the same, so the semantics is preserved:
\begin{align}
    \mathsf{d.f(this \mapsto d).\varphi} &\redtomany \mathsf{2.add(n \mapsto d.g(this \mapsto d).\varphi)} \\
    & \redtomany \mathsf{2.add(n \mapsto d.\varphi.g(this \mapsto d).\varphi)}\\
    & \redtomany \mathsf{2.add(n \mapsto Base.new.g(this \mapsto d).\varphi)} \label{eq:dynamic-dispatch}\\
    & \redtomany \mathsf{2.add(n \mapsto d.h(this \mapsto d).\varphi)}\\
    &\redtomany \mathsf{2.add(n \mapsto 5)}
\end{align}

Note that in (\ref{eq:dynamic-dispatch}), `method' $\mathsf{g}$ is called from the $\mathsf{Base}$ `class' with $(\mathsf{this \mapsto d})$. This is what provides proper support of the dynamic dispatch (open recursion) in modeling of classes in $\varphi$-calculus.

\end{example}

The above example aims to provide intuition for $\varphi$-calculus. A proper mapping from Java to $\varphi$-calculus would require more technical details, such as dealing with mutable attributes, generics, interfaces, and other features.

\subsubsection{Prototype-based}

Prototypes in object-oriented languages, such as JavaScript, work similarly to decorators in $\varphi$-calculus: when looking for a method in a JavaScript object, the interpreter checks object's \emph{own properties} first and then, if such properties are absent, the interpreter proceeds to look for the method in the object's \emph{prototype}, unless the prototype is \emph{null}. Importantly, JavaScript's objects are mutable allowing for dynamic prototypes and properties, whereas in $\varphi$-calculus objects are immutable.

\begin{example}
Consider the following code snippet in JavaScript:

\begin{minted}{javascript}
let A = function() { this.x = 3; }
A.prototype.f = function() { return this.x; }
let b = new A();
let c = b.f()
\end{minted}

This snippet would translate to $\varphi$-calculus as follows:
\begin{align*}
  \mathsf{A} := & \mkObject{
    \mathsf{new} \mapsto \mkObject{
      \mathsf{x} \mapsto 3,
      \varphi \mapsto \rho^1.\mathsf{prototype}
    }, \\
    &\;\;\mathsf{prototype} \mapsto \mkObject{
      \mathsf{f} \mapsto \mkObject{
        \mathsf{this} \mapsto \varnothing,
        \varphi \mapsto \rho^0.\mathsf{this}.\mathsf{x}
      }
    }
  } \\
  \mathsf{b} := & \mathsf{A}.\mathsf{new}\\
  \mathsf{c} := &\mathsf{b}.f (\mathsf{this} \mapsto \mathsf{b})
\end{align*}

Similarly to the translation of Java's class-based constructs, here $\mathsf{b}$ is passed as $\mathsf{this}$ argument in the method call (application of) $\mathsf{b.f}$ in order to preserve information about an object that originally called $f$; this is essential whenever attributes are resolved in decoratees (higher in the hierarchy of objects, representing prototypes).
\end{example}

%% file: figures/core-calculus.tex
\begin{figure*}
  \begin{mdframed}

    \emph{Syntax}
    \begin{align*}
        t &:=  \tag{terms} \\
          &\phantom{\mid\;\;}   t.a \tag{attribute} \\
          &\mid   t_1(a \mapsto t_2) \tag{application} \\
          &\mid   \mkObject{ a_1 \mapsto \varnothing, \ldots, a_k \mapsto \varnothing, b_1 \mapsto t_1, \ldots, b_n \mapsto t_n } \tag{object} \\
          &\mid   \rho^n \tag{$n$-th parent object locator}
    \end{align*}
    \hfill\hrule\hfill

    \emph{Evaluation}
    \begin{prooftree}
      \AxiomC{$t_i \rightsquigarrow t_i'$}
      \RightLabel{cong$_{\text{OBJ}}$}
      \UnaryInfC{$\mkObject{\ldots, b_i \mapsto t_i, \ldots} \rightsquigarrow \mkObject{\ldots, b_i \mapsto t_i', \ldots}$}
      \DisplayProof\quad\quad
      \AxiomC{$t \rightsquigarrow t'$}
      \RightLabel{cong$_{\text{DOT}}$}
      \UnaryInfC{$t.a \rightsquigarrow t'.a$}
    \end{prooftree}
    \begin{prooftree}
      \AxiomC{$t \rightsquigarrow t'$}
      \RightLabel{cong$_{\text{APP}^\text{L}}$}
      \UnaryInfC{$t(a \mapsto u) \rightsquigarrow t'(a \mapsto u)$}
      \DisplayProof\quad\quad
      \AxiomC{$u \rightsquigarrow u'$}
      \RightLabel{cong$_{\text{APP}^\text{R}}$}
      \UnaryInfC{$t(a \mapsto u) \rightsquigarrow t'(a \mapsto u')$}
    \end{prooftree}
    \begin{prooftree}
      \AxiomC{$t \equiv \mkObject{\ldots, \textcolor{black}{c \mapsto t_c}, \ldots}$}
      \RightLabel{DOT$_{c}$}
      \UnaryInfC{$t.c \rightsquigarrow \textcolor{black}{t_c} \left[ \rho^0 \mapsto t \right]$}
      \DisplayProof\quad\quad\quad
      \AxiomC{$t \equiv \mkObject{\ldots}$}
      \AxiomC{$c \notin \attr(t)$}
      \AxiomC{$\varphi \in \attr(t)$}
      \RightLabel{DOT$_{c}^\varphi$}
      \TrinaryInfC{$t.c \rightsquigarrow t.\varphi.c$}
    \end{prooftree}
    \begin{prooftree}
      \AxiomC{$t \equiv \mkObject{a_1 \mapsto \varnothing, \ldots, a_k \mapsto \varnothing, \textcolor{black}{c \mapsto \varnothing}, b_1 \mapsto t_1, \ldots, b_n \mapsto t_n}$}
      \RightLabel{APP$_{c}$}
      \UnaryInfC{$t (c \mapsto u) \rightsquigarrow \mkObject{a_1 \mapsto \varnothing, \ldots, a_k \mapsto \varnothing, \textcolor{black}{c \mapsto \inct{}{u}}, b_1 \mapsto t_1, \ldots, b_n \mapsto t_n}$}
    \end{prooftree}
    \hfill\hrule\hfill

    \emph{Locator substitution}
    \begin{align*}
        \rho^n[\rho^m \mapsto u] &:= \rho^n \quad\text{if $n < m$} \\
        \rho^n[\rho^n \mapsto u] &:= u \\
        \rho^n[\rho^m \mapsto u] &:= \rho^{n-1} \quad\text{if $n > m$} \\
        t.a[\rho^n \mapsto u] &:= t[\rho^n \mapsto u].a \\
        t_1(a \mapsto t_2)[\rho^n \mapsto u] &:= t_1[\rho^n \mapsto u](a \mapsto t_2[\rho^n \mapsto u]) \\
        \mkObject{a_1 \mapsto \varnothing, \ldots, b_1 \mapsto t_1, \ldots}[\rho^n \mapsto u]
            &:= \mkObject{a_1 \mapsto \varnothing, \ldots, b_1 \mapsto t_1[\rho^{n+1} \mapsto \inct{}{u}], \ldots}
    \end{align*}
  \end{mdframed}

    \caption{Syntax and evaluation rules for $\varphi$-calculus.}
    \label{fig:core-phi-calculus}
\end{figure*}

%% file: sections/3-confluence.tex
\section{Confluence}

Intuitively, we think of $\varphi$-terms as programs in terms of objects. Moreover, we assume a single meaning to each such program. In other words, every program either diverges (e.g. falls into an infinite loop) or produces the final object (normal form). To justify the use of ``the'' in ``the normal form'' we require the uniqueness of normal form.

In this section we prove a more general result. We show that $\varphi$-calculus possesses the following property: if some term $t$ can be reduced in different ways to terms $u$ and $v$ then there exists some term $w$ such that both $u$ and $v$ reduce to $w$.
This is known as Church-Rosser property:

\begin{definition}
    A relation $\xrightarrow[]{X}$ on terms is said to satisfy \emph{Church-Rosser property} if for any terms $t$, $u$ and $v$, if $t \xrightarrow[]{X} u$ and $t \xrightarrow[]{X} v$, then there exists some term $w$ such that $u \xrightarrow[]{X} w$ and $v \xrightarrow[]{X} w$.
\end{definition}

In general, a $\varphi$-term can be rewritten in different ways, since it may contain several redexes. For example, here is a graph with all possible reductions for a term:
\[
\begin{tikzcd}[column sep=-17ex]
    & \mkObject{x \mapsto \mkObject{y \mapsto \varnothing}}.x(y \mapsto \mkObject{z \mapsto w}.z) \ar[dl] \ar[dr] & \\
    \mkObject{y \mapsto \varnothing}(y \mapsto \mkObject{z \mapsto w}.z)
    \ar[d] \ar[drr]
        & & \mkObject{x \mapsto \mkObject{y \mapsto \varnothing}}.x(y \mapsto w) \ar[d] \\
    \mkObject{y \mapsto \mkObject{z \mapsto w}.z} \ar[dr]
        & & \mkObject{y \mapsto \varnothing}(y \mapsto w) \ar[dl] \\
    & \mkObject{y \mapsto w} & 
\end{tikzcd}
\]

Other terms may have infinite rewrite sequences, e.g.:
\[
\begin{tikzcd}
    \mkObject{x \mapsto \rho^0.y, y \mapsto \rho^0.x}.x
    \ar[d, bend right] \\
    \mkObject{x \mapsto \rho^0.y, y \mapsto \rho^0.x}.y
    \ar[u, bend right]
\end{tikzcd}
\]

In these two examples, we can see diamond property satisfied for $\rightsquigarrow$, but unfortunately, the property does not hold in general.

\begin{example}
    \label{example:diamond-property-counterexample}
    Consider the following $\varphi$-terms $A$ and $B$:
    \begin{alignat*}{3}
        A &\equiv \mkObject{x \mapsto \mkObject{a \mapsto 
        & &\mkObject{z \mapsto \varnothing}}.a
        &, y \mapsto \rho^0.x(z \mapsto \rho^0.x)} \\
        B &\equiv \mkObject{x \mapsto 
        & &\mkObject{z \mapsto \varnothing} 
        &, y \mapsto \rho^0.x(z \mapsto \rho^0.x)}
    \end{alignat*}
    Note that $A \redto B$ by $\text{cong}_{\text{OBJ}}$ and $\text{DOT}_a$. Reduction of $A.y$ illustrates that substitution can introduce multiple redexes, that cannot be reduced in a single step of $\redto$:
    \[
    \begin{tikzcd}
        & A.y \ar[dl, bend right=7] \ar[drr, bend left=7] &\\
        B.y \ar[ddr, bend right=30]
        & & & A.x(z \mapsto A.x) \ar[dll, bend right=20] \ar[d] \ar[ddll, dashed, "\text{\textcolor{red}{no $\redto$ here}}" description]\\
        & B.x(z\mapsto A.x) \ar[d]
        & & B.x(z\mapsto A.x) \ar[dll, bend left=20] \\
        & B.x(z\mapsto B.x) & & 
    \end{tikzcd}
    \]
\end{example}

To prove confluence for $\varphi$-calculus we follow these steps:
\begin{enumerate}
    \item We introduce parallel reduction on $\varphi$-terms; this kind of reduction possesses the diamond property. Using parallel reduction in a proof of confluence is due to Tait and L\"of \cite[Section 3.2]{Barendregt1984}
    \item We show that parallel reduction is equivalent to regular reduction.
    \item We show that parallel reduction possesses the diamond property via complete development, following Takahashi's technique \cite{Takahashi1995}.
    \item We prove confluence for regular reduction via equivalence with parallel reduction, also using the fact that confluence for parallel reduction follows from its diamond property via \cite[Lemma 1.17]{Krivine1993}.
\end{enumerate}

While regular reduction performs exactly one reduction step somewhere in a term, the idea of parallel reduction is to perform arbitrary number of reductions \emph{in parallel}. Performing reductions in parallel intuitively means that we do not reduce a term after performing substitution or adding $.\varphi$. Figure~\ref{fig:parallel-reduction} gives the rules for parallel reduction.
Observe that $\varphi$-terms in Example~\ref{example:diamond-property-counterexample} satisfy the diamond property with single step parallel reduction: $A.x(z \mapsto A.x) \Rrightarrow B.x(z \mapsto B.x)$ by cong$_{\text{APP}}^{\Rrightarrow}$.

\input{figures/core-calculus-parallel}

One important property of parallel reduction is that ``doing nothing'' is also a parallel reduction. We justify this with the following proposition:

\begin{proposition}[Reflexivity of parallel reduction]
    Let $t$ be a $\varphi$-term. Then $t \Rrightarrow t$.
\end{proposition}
\begin{proof}
    Straightforward by structural induction on $t$.
\end{proof}

Since our goal is to prove confluence for regular reduction via parallel reduction, we need to establish that the two kinds of reduction are equivalent, meaning that if one term reduces regularly to another term, then those terms are also related via parallel reduction and vice versa.

\begin{proposition}[Equivalence of $\Rrightarrow$ and $\rightsquigarrow$]
\label{proposition:equivalence-parallel-regular-reduction}
    Parallel reduction ($\Rrightarrow$) is equivalent to regular reduction ($\rightsquigarrow$):
    \begin{enumerate}
        \item $t \rightsquigarrow t'$ implies $t \Rrightarrow t'$
        \item $t \redtomany t'$ implies $t \parredtomany t'$
        \item $t \Rrightarrow t'$ implies $t \redtomany t'$
        \item $t \parredtomany t'$ implies $t \redtomany t'$
    \end{enumerate}
\end{proposition}
\begin{proof}
    Straightforward by structural induction.
\end{proof}

In the following proofs we need to know that if $t \Rrightarrow t'$ and $u \Rrightarrow u'$ then $t[\rho^0 \mapsto u] \Rrightarrow t'[\rho^0 \mapsto u']$. We prove a slightly more general lemma:

\begin{lemma}[Substitution lemma]
    Let $t, t', u, u'$ be $\varphi$-terms and $t \Rrightarrow t'$ and $u \Rrightarrow u'$.
    Then $t[\rho^i \mapsto u] \Rrightarrow t'[\rho^i \mapsto u']$.
\end{lemma}

To show that parallel reduction satisfies the diamond property, we adapt Takahashi's technique \cite{Takahashi1995} and define complete development of a term, which is intuitively the maximum possible one-step parallel reduction of a term. The idea is that if $t \Rrightarrow t'$ then simply by performing all parallel reductions that were ``skipped'' when producing $t'$, we get from $t'$ to the complete development of $t$. More formally:

\begin{definition}
    Let $t$ be a $\varphi$-term. Then a term $t^{+}$ denotes the \emph{complete development} of $t$, defined recursively as follows:
    \begin{align*}
        (\rho^n)^{+}
          &:= \rho^n \\
        (t.a)^{+}
          &:=
          \begin{cases}
            t_a[\rho^0 \mapsto t^{+}], & \text{if $t^{+} \equiv \mkObject{\ldots, a \mapsto t_a, \ldots}$} \\
            t^{+}.\varphi.a, & \text{if $a \notin \attr(t^{+})$ and $\varphi \in \attr(t^{+})$} \\
            t^{+}.a & \text{otherwise}
          \end{cases} \\
        (t(a \mapsto u))^{+}
          &:=
          \begin{cases}
            \mkObject{a \mapsto \inct{}{u^{+}}, \ldots}, & \text{if $t^{+} \equiv \mkObject{a \mapsto \varnothing, \ldots}$} \\
            t^{+}(a \mapsto u^{+}) & \text{otherwise}
          \end{cases} \\
        (\mkObject{a_1 \mapsto \varnothing, \ldots, a_k \mapsto \varnothing, b_1 \mapsto t_1, \ldots, b_n \mapsto t_n})^{+} \span \\
          := \mkObject{a_1 \mapsto \varnothing, \ldots, a_k \mapsto \varnothing, b_1 \mapsto t_1^{+}, \ldots, b_n \mapsto t_n^{+}} \span
    \end{align*}
\end{definition}

The definition clearly shows that we basically follow rules of parallel reduction, but always choose the rule that does the most work. In particular, if both $\text{APP}_c^{\Rrightarrow}$ and $\text{cong}_{\text{APP}}^{\Rrightarrow}$ are applicable, we prioritize rule $\text{APP}_c^{\Rrightarrow}$ since it performs strictly more reductions. Consequently, a term can always be parallel reduced to its complete development:

\begin{proposition}
    Let $t$ be a $\varphi$-term. Then $t \Rrightarrow t^{+}$.
\end{proposition}
\begin{proof}
    Straightforward by induction on $t$.
\end{proof}

Before we prove the diamond property, we need a simpler result for just one half of the diamond:

\begin{proposition}
    \label{prop:reduction-to-cd}
    Let $t, t'$ be $\varphi$-terms and $t \Rrightarrow t'$. Then $t'~\Rrightarrow~t^{+}$.
\end{proposition}

\begin{corollary}[Diamond property of parallel reduction]
    Let $t, u, v$ be $\varphi$-terms and $t \Rrightarrow u$ and $t \Rrightarrow v$. Then there exists a $\varphi$-term $w$ such that $u \Rrightarrow w$ and $v \Rrightarrow w$.
\end{corollary}
\begin{proof}
    Let $w\equiv t^+$. With \ref{prop:reduction-to-cd}, $u \Rrightarrow w$ and $v \Rrightarrow w$.
\end{proof}

\begin{corollary}[Confluence of parallel reduction]
\label{corollary:parallel-confluence}
    Let $t, u, v$ be $\varphi$-terms and $t \parredtomany u$ and $t \parredtomany v$. Then there exists a $\varphi$-term $w$ such that $u \parredtomany w$ and $v \parredtomany w$.
\end{corollary}
\begin{proof}
    Follows from the diamond property (see \cite[Lemma 1.17]{Krivine1993}).
\end{proof}

\begin{theorem}[Confluence]
\label{theorem:confluence}
    Let $t, u, v$ be $\varphi$-terms and $t \redtomany u$ and $t \redtomany v$. Then there exists a $\varphi$-term $w$ such that $u \redtomany w$ and $v \redtomany w$.
\end{theorem}
\begin{proof}
    Follows from Proposition~\ref{proposition:equivalence-parallel-regular-reduction} and Corollary~\ref{corollary:parallel-confluence}.
\end{proof}

As usual, a term $t$ is in \emph{normal form} if it has no redexes. In $\varphi$-calculus, it means:
\begin{definition}
  $\varphi$-term $t$ is said to have the normal form if 
  \begin{align*}
        t \equiv
        \begin{cases}
            \rho^n\\
            \mkObject{a_1 \mapsto \varnothing, \ldots, a_k \mapsto \varnothing, b_1 \mapsto t_1, \ldots, b_n \mapsto t_n}, &\text{if $t_j$ is in NF}\\
            s.a, &\text{if $s$ is in NF and $s.a$ is not a redex}\\
            s(a\mapsto u), &\text{if $s$ and $u$ are in NF and $s(a\mapsto u)$ is not a redex}\\
        \end{cases} 
    \end{align*}
\end{definition}

With this definition we have a trivial corollary of Theorem~\ref{theorem:confluence}:

\begin{corollary}
    Every $\varphi$-term has at most one normal form.
\end{corollary}


For the reasoning about the abstract machine in the section \ref{section:abstract-machine}, weak head normal form needs to be defined:
\begin{definition}
\label{def:weak-head-normal-form}
  $\varphi$-term $t$ is said to be in a weak head normal form if 
  \begin{align*}
        t \equiv
        \begin{cases}
            \rho^n\\
            \mkObject{a_1 \mapsto \varnothing, \ldots, a_k \mapsto \varnothing, b_1 \mapsto t_1, \ldots, b_n \mapsto t_n}\\
            s.a, &\text{if $s$ is in WHNF and $s.a$ is not a redex}\\
            s(a\mapsto u), &\text{if $s$ is in WHNF and $s(a\mapsto u)$ is not a redex}\\
        \end{cases} 
    \end{align*}
    Equivalently, one might have said that $t$ is in a weak head normal if no head reductions (defined in the figure \ref{figure:normal-reduction}) starting from $t$ are possible.
\end{definition}

\begin{example}
    Term $\mkObject{z \mapsto \mkObject{x \mapsto \rho^0.y, y \mapsto \rho^0.x}.x}$ is in weak head normal form, but it does not have a normal form.
\end{example}

\subsection{Completeness of normal order evaluation}

In this subsection, we define normal order reduction strategy and prove that this strategy reduces a term to its normal form, if such a form exists.

Similarly to $\lambda$-calculus, normal order reduction in $\varphi$-calculus is defined in terms of head reduction. Head reduction does not apply to redexes inside objects and to arguments of object application. If possible, normal order reduction performs head reduction, otherwise it performs reduction in the leftmost subterm that can be reduced.

Our proof of completeness of normal order evaluation relies on the standardization theorem: if $t$ reduces to $u$, then there exists a reduction path, where head reductions are performed first, and internal reductions follow.

Contrary to the head reduction, internal reduction applies to redexes inside objects and in arguments of application. During the proof, we exploit internal parallel reduction which is the intersection of parallel reduction and reflexive-transitive closure of the internal regular reduction.

\input{figures/normal-reduction}

The proof is developed in a close relation to the one of Takahashi \cite{Takahashi1995}.
First, we show that one step of parallel reduction can be decomposed to multiple head reductions and one internal parallel reduction step.

\begin{lemma}[Main Lemma]
    \label{lemma:normal-reduction:main}
    $t \Rrightarrow s$ implies $t \headredmany r \innerparred s$ for some $r$.
\end{lemma}

\begin{lemma}[Substitution Lemma for $\headred$]
    \label{lemma:normal-reduction:substitution}
    If $t \headred s$, then $t[\rho^n \mapsto q] \headred s[\rho^n \mapsto q]$.
\end{lemma}

\begin{lemma}[Substitution Lemma for $\innerparred$]
    \label{lemma:normal-reduction:substitution-inner}
    If $t \innerparred s$ and $q \innerparred r$, then $t[\rho^n \mapsto q] \innerparred s[\rho^n \mapsto r]$.
\end{lemma}

Then, we show that if head reduction follows internal parallel reduction, reductions can be reordered so that head reductions occur first.
\begin{lemma}[Standardizing Reductions]
    \label{lemma:normal-reduction:standardizing}
    For any $\varphi$-terms $t, r, s$ such that $t \innerparred r \headred s$, there exists $\varphi$-term $q$, such that $t \headredmany q \innerparred s$.
\end{lemma}

\begin{proof}
By induction on the structure of $r \headred s$.
\end{proof}

Standardization theorem is then a corollary:
\begin{corollary}
    $t \redtomany s$ implies $t \headredmany r \innerredmany s$ for some $\varphi$-term $r$.
\end{corollary}
    
\begin{proof}
    Recall that equivalence of $\redtomany$ and $\parredtomany$ (4) implies that $t \parredtomany s$.
    
    By induction on $\parredtomany$,
    \begin{enumerate}
        \item if $t \equiv s$, then $t \headredmany r \innerredmany s$ for $r \equiv s$
        \item else, $t \Rrightarrow q$ and $q \parredtomany s$. By Main Lemma, there exists $p$, such that $t \headredmany p \innerparred q$. By induction hypothesis, there exists $r'$, such that $q \headredmany r' \innerredmany s$. Repeated application of the Standardizing Reductions Lemma propagates $\innerparred$ in $p\innerparred q \headredmany r'$ to the end, and the equivalence of $\innerparred$ and $\innerredmany$ completes the proof.
    \end{enumerate}
\end{proof}

\begin{theorem}[Completeness of normal order evaluation]
    If $t$ has normal form $s$, then $t \leftmostmany s$.
\end{theorem}
\begin{proof}
With the corollary, $t \headredmany r \innerredmany s$ for some $\varphi$-term $r$.

By induction on the structure of $s$,
\begin{enumerate}
    \item if $s \equiv \rho^n$, then $r \equiv \rho^n$. As $t \headredmany r$, $t \leftmostmany r$, which follows from the definition of $\leftmost$.
    \item if $s \equiv \mkObject{a_1 \mapsto \varnothing, \ldots, a_k \mapsto \varnothing, b_1 \mapsto t_1, \ldots, b_n \mapsto t_n}$, the $r \equiv \mkObject{a_1 \mapsto \varnothing, \ldots, a_k \mapsto \varnothing, b_1 \mapsto t_1', \ldots, b_n \mapsto t_n'}$ and $t_j$ is NF of $t_j'$. By induction hypothesis, $t_j' \leftmostmany t_j$, so $r \leftmostmany s$, hence $t\leftmostmany s$.
    \item if $s \equiv q.a$, then $q$ is in NF and $q.a$ is not a redex, hence $r \equiv q'.a$, and $q$ is NF of $q'$. By induction hypothesis, $q' \leftmostmany q$, and since $q.a$ is not a redex, $q'.a \leftmostmany q.a$. So, $t\headredmany r\leftmostmany s$, and by definition of $\leftmost$, $t\leftmostmany s$.
    \item if $s \equiv q(a\mapsto u)$, then $q$ and $u$ are in NF and $q(a\mapsto u)$ is not a redex, hence $r \equiv q'(a\mapsto u')$, and $q$ is NF of $q'$ and $u$ is NF of $u'$. By induction hypothesis, $q' \leftmostmany q$ and $u' \leftmostmany u$, and since $q(a\mapsto u)$ is not a redex, $q'(a\mapsto u') \leftmostmany q(a\mapsto u') \leftmostmany q(a\mapsto u)$. So, $t\leftmostmany s$.
\end{enumerate}
\end{proof}

%% file: figures/core-calculus-parallel.tex
\begin{figure*}
  \begin{mdframed}

    \begin{prooftree}
      \AxiomC{$t_1 \Rrightarrow t_1'\quad\ldots\quad t_n \Rrightarrow t_n'$}
      \RightLabel{cong$_{\text{OBJ}}^{\Rrightarrow}$}
      \UnaryInfC{$\mkObject{a_i \mapsto \varnothing, b_j \mapsto t_j} \Rrightarrow \mkObject{a_i \mapsto \varnothing, b_j \mapsto t_j'}$ for $i \in \{1, \ldots, k\}, j \in \{1, \ldots, n\}$}
    \end{prooftree}
    
    \begin{prooftree}
      \AxiomC{}
      \RightLabel{cong$_{\rho}^{\Rrightarrow}$}
      \UnaryInfC{$\rho^n \Rrightarrow \rho^n$}
      \DisplayProof\quad\quad
      \AxiomC{$t \Rrightarrow t'$}
      \RightLabel{cong$_{\text{DOT}}^{\Rrightarrow}$}
      \UnaryInfC{$t.a \Rrightarrow t'.a$}
      \DisplayProof\quad\quad
      \AxiomC{$t \Rrightarrow t'$}
      \AxiomC{$u \Rrightarrow u'$}
      \RightLabel{cong$_{\text{APP}}^{\Rrightarrow}$}
      \BinaryInfC{$t(a \mapsto u) \Rrightarrow t'(a \mapsto u')$}
    \end{prooftree}

    \begin{prooftree}
      \AxiomC{$t \Rrightarrow t' \quad t' \equiv \mkObject{\ldots, \textcolor{black}{c \mapsto t_c}, \ldots}$}
      \RightLabel{DOT$_{c}^{\Rrightarrow}$}
      \UnaryInfC{$t.c \Rrightarrow \textcolor{black}{t_c} \left[ \rho^0 \mapsto t' \right]$}
      \DisplayProof\quad\quad\quad
      \AxiomC{$t \Rrightarrow t' \quad t' \equiv \mkObject{\ldots}$}
      \AxiomC{$c \notin \attr(t') \quad \varphi \in \attr(t')$}
      \RightLabel{DOT$_{c}^{\varphi\Rrightarrow}$}
      \BinaryInfC{$t.c \Rrightarrow t'.\varphi.c$}
    \end{prooftree}

    \begin{prooftree}
      \AxiomC{$t \Rrightarrow t' \quad t' \equiv \mkObject{a_1 \mapsto \varnothing, \ldots, a_k \mapsto \varnothing, \textcolor{black}{c \mapsto \varnothing}, b_1 \mapsto t_1, \ldots, b_n \mapsto t_n}$}
      \AxiomC{$u \Rrightarrow u'$}
      \RightLabel{APP$_{c}^{\Rrightarrow}$}
      \BinaryInfC{$t (c \mapsto u) \Rrightarrow \mkObject{a_1 \mapsto \varnothing, \ldots, a_k \mapsto \varnothing, \textcolor{black}{c \mapsto \inct{}{u'}}, b_1 \mapsto t_1, \ldots, b_n \mapsto t_n}$}
    \end{prooftree}
  \end{mdframed}

    \caption{Parallel reduction rules for $\varphi$-calculus.}
    \label{fig:parallel-reduction}
\end{figure*}

%% file: figures/normal-reduction.tex
\begin{figure*}
  \begin{mdframed}
  
  \emph{Head reduction}
    \begin{prooftree}
      \AxiomC{$t \headred t'$}
      \RightLabel{cong$_{\text{DOT}}^h$}
      \UnaryInfC{$t.a \headred t'.a$}
      \DisplayProof\quad\quad\quad
      \AxiomC{$t \headred t'$}
      \RightLabel{cong$_{\text{APP}}^h$}
      \UnaryInfC{$t(a \mapsto u) \headred t'(a \mapsto u)$}
    \end{prooftree}
    \begin{prooftree}
      \AxiomC{$t \equiv \mkObject{\ldots, \textcolor{black}{c \mapsto t_c}, \ldots}$}
      \RightLabel{DOT$_{c}$}
      \UnaryInfC{$t.c \headred t_c \left[ \rho^0 \mapsto t \right]$}
      \DisplayProof\quad\quad\quad
      \AxiomC{$t \equiv \mkObject{\ldots}$}
      \AxiomC{$c \notin \attr(t)$}
      \AxiomC{$\varphi \in \attr(t)$}
      \RightLabel{DOT$_{c}^\varphi$}
      \TrinaryInfC{$t.c \headred t.\varphi.c$}
    \end{prooftree}
    \begin{prooftree}
      \AxiomC{$t \equiv \mkObject{a_1 \mapsto \varnothing, \ldots, a_k \mapsto \varnothing, \textcolor{black}{c \mapsto \varnothing}, b_1 \mapsto t_1, \ldots, b_n \mapsto t_n}$}
      \RightLabel{APP$_{c}$}
      \UnaryInfC{$t (c \mapsto u) \headred \mkObject{a_1 \mapsto \varnothing, \ldots, a_k \mapsto \varnothing, \textcolor{black}{c \mapsto \inct{}{u}}, b_1 \mapsto t_1, \ldots, b_n \mapsto t_n}$}
    \end{prooftree}

    \hfill\hrule\hfill

    \emph{Normal order}
    \begin{align*}
        t.a
          &\leftmost
          \begin{cases}
            s, & \text{if $t.a \headred s$} \\
            t'.a, & \text{else if $t \leftmost t'$}
          \end{cases} \\
        t(a \mapsto u)
          &\leftmost
          \begin{cases}
            s, & \text{if $t (a\mapsto u) \headred s$} \\
            t'(a \mapsto u), & \text{else if $t \leftmost t'$} \\
            t(a \mapsto u'), & \text{else if $u \leftmost u'$}
          \end{cases} \\
        \mkObject{a_i \mapsto \varnothing, b_j \mapsto t_j}
          &\leftmost \mkObject{a_i \mapsto \varnothing, b_j \mapsto t_j'}, \text{  if $t_j \leftmost t_j'$}
    \end{align*}
    \hfill\hrule\hfill

    \emph{Regular internal reduction}
    \begin{prooftree}
      \AxiomC{$t_i \rightsquigarrow t_i'$}
      \RightLabel{cong$^i_{\text{OBJ}}$}
      \UnaryInfC{$\mkObject{\ldots, b_i \mapsto t_i, \ldots} \innerred \mkObject{\ldots, b_i \mapsto t_i', \ldots}$}
      \DisplayProof\quad\quad
      \AxiomC{$t \innerred t'$}
      \RightLabel{cong$^i_{\text{DOT}}$}
      \UnaryInfC{$t.a \innerred t'.a$}
    \end{prooftree}
    \begin{prooftree}
      \AxiomC{$t \innerred t'$}
      \RightLabel{cong$^i_{\text{APP}^\text{L}}$}
      \UnaryInfC{$t(a \mapsto u) \innerred t'(a \mapsto u)$}
      \DisplayProof\quad\quad
      \AxiomC{$u \rightsquigarrow u'$}
      \RightLabel{cong$^i_{\text{APP}^\text{R}}$}
      \UnaryInfC{$t(a \mapsto u) \innerred t'(a \mapsto u')$}
    \end{prooftree}
    
    \hfill\hrule\hfill
    
    \emph{Parallel internal reduction}
    \begin{prooftree}
      \AxiomC{$t_1 \Rrightarrow t_1'\quad\ldots\quad t_n \Rrightarrow t_n'$}
      \RightLabel{cong$_{\text{OBJ}}^{\innerparred}$}
      \UnaryInfC{$\mkObject{a_i \mapsto \varnothing, b_j \mapsto t_j} \innerparred \mkObject{a_i \mapsto \varnothing, \ldots, b_j \mapsto t_j'}$}
    \end{prooftree}
    \begin{prooftree}
      \AxiomC{$\phantom\innerparred$}
      \RightLabel{cong$_{\rho}^{\innerparred}$}
      \UnaryInfC{$\rho^n \innerparred \rho^n$}
      \DisplayProof\quad\quad
      \AxiomC{$t \innerparred t'$}
      \RightLabel{cong$_{\text{DOT}}^{\innerparred}$}
      \UnaryInfC{$t.a \innerparred t'.a$}
      \DisplayProof\quad\quad
      \AxiomC{$t \innerparred t'$}
      \AxiomC{$u \Rrightarrow u'$}
      \RightLabel{cong$_{\text{APP}}^{\innerparred}$}
      \BinaryInfC{$t(a \mapsto u) \innerparred t'(a \mapsto u')$}
    \end{prooftree}

  \end{mdframed}

\caption{Head, internal and normal order reductions}
\label{figure:normal-reduction}
\end{figure*}

%% file: sections/4-abstract-machine.tex
\section{Abstract machine}
\label{section:abstract-machine}

Bugayenko \cite{bugayenko/online} gives graph-based operational semantics. Although the general idea is clear, his description lacks precision, in particular regarding the handling of parent locators. Still, the existing implementations of EO programming language and descriptions of graph-based semantics hint strongly that intended semantics are those of a non-strict evaluation.

In this section, we introduce an abstract machine \`a la Krivine that performs call-by-name reduction of $\varphi$-terms, therefore computing their weak head normal form (defined in \ref{def:weak-head-normal-form}). 

\subsection{Call-by-name abstract machine}

Here we present \emph{term-actions-parents} abstract machine (TAP machine) for call-by-name evaluation of $\varphi$-terms. 

We begin by introducing configurations of TAP machine:
\begin{definition}
    An \emph{object closure} is a tuple $(t, e)$ of a $\varphi$-term $t$ and a parent stack, described below.
    A \emph{parent} is a tuple $(t, o)$ of an object $\varphi$-term $t$ and a partial mapping $o$ (called \emph{application mapping}) from $\mathcal{L}$ to a set of object closures.
    A \emph{parent stack} is a finite sequence of parents. An empty parent stack is denoted $\epsilon$.
    An \emph{action} is either an attribute access denoted $.a$ for some attribute $a \in \mathcal{L}$, or an application denoted $(a \mapsto c)$ for some attribute $a \in \mathcal{L}$ and an object closure $c$. An \emph{action stack} is a finite sequence of actions. An empty action stack is denoted $\epsilon$.
    A \emph{configuration} of TAP machine is a triple $\langle T, A, P \rangle$, where $T$ is either a $\varphi$-term in focus or an empty symbol $\epsilon$, $A$ is a stack of actions, and $P$~--- a parent stack.
\end{definition}

The machine operates by following transition rules between the configurations. Figure~\ref{figure:tap-machine} gives the transition rules. The first two rules instruct how to dereference parent locator $\rho^n$. Attribute access and application terms are broken down into a smaller term and an action. For application $t(a \mapsto u)$ we save the current parent stack $e$ and put an action $(a \mapsto (u, e))$ on the stack of actions. This effectively captures the necessary context required to compute term $u$ later. Rule~\ref{rule:tap-from-object} puts an object term on the stack with an empty application mapping (denoted $\varnothing$). The remaining four rules describe effects of actions on the parent on the top of the stack. If the parent object on the stack has an attribute that we want to access, we extract the corresponding  subterm and set it as our new current term. On the other hand, if the attribute is mapped by the parent application mapping to some object closure, then we take the term from that closure as our new current term and replace current parent stack with the one from the closure. If a parent has no required attribute, but has $\varphi$, we simply add $.\varphi$ action to the action stack. Finally, an application action merely updates the parent at the top of the parent stack, by replacing its application mapping correspondingly: $o \cup \{a \mapsto (u, e')\}$ denotes a mapping that maps attribute $a$ to object closure $(u, e')$ and maps any other attribute $x$ to $o(x)$.

\input{figures/tap-machine}

An object closure can be converted back to a $\varphi$-term by converting every parent into a $\varphi$-term and then performing locator substitution, instantiating corresponding parents.
A parent $(t, o)$ is converted into a $\varphi$-term by joining its object term $t$ with its application mapping and converting every object closure in that mapping to a $\varphi$-term.
Any configuration $\langle t, A, P \rangle$ (or $\langle \epsilon, A, p:P \rangle$) can be converted back to a $\varphi$-term by appending the stack of actions, where object closures are converted to $\varphi$-terms, to the term produced from the closure $(t, P)$ (resp. $(t, P)$ where $t$ is produced from $p$).

\begin{proposition}[Soundness of TAP machine]
    Let $t$ be a closed $\varphi$-term. Then starting from configuration $C_0 = \mkState{t, \epsilon, \epsilon}$ TAP machine operates for finitely many steps if and only if $t$ has a weak head normal form. Moreover, if it stops with configuration $C_n$ then this configuration corresponds to the weak head normal form of $t$.
\end{proposition}
\begin{proof}
  Each reduction in call-by-name evaluation sequence corresponds unambiguously to zero or more transitions of the TAP machine. Transitions from equivalent configurations, corresponding to the same $\varphi$-term, destructure current term, so there can only be a finite sequence of them.
\end{proof}

%% file: figures/tap-machine.tex
\begin{figure*}
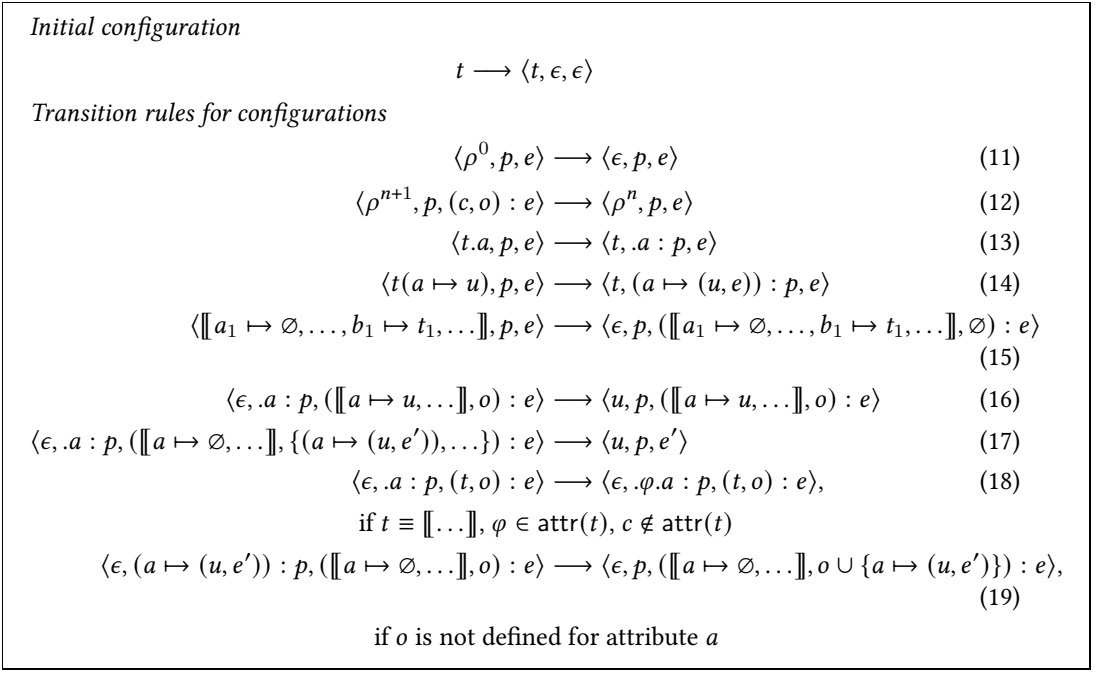

  \begin{mdframed}
 \emph{Initial configuration}
 \begin{align*}
 t \longrightarrow \mkState{t, \epsilon, \epsilon}
 \end{align*}
 
 \emph{Transition rules for configurations}
 \begin{align}
    \mkState{\rho^0, p, e} &\longrightarrow \mkState{\epsilon, p, e }\\
    \mkState{\rho^{n+1}, p, (c,o):e } &\longrightarrow \mkState{\rho^n, p, e }\\
    \mkState{t.a, p, e } &\longrightarrow \mkState{t, .a:p, e }\\
    \mkState{t(a \mapsto u), p, e } &\longrightarrow \mkState{t, (a \mapsto (u, e)):p, e }\\
    \mkState{\mkObject{a_1 \mapsto \varnothing, \ldots, b_1 \mapsto t_1, \ldots}, p, e } &\longrightarrow
    \mkState{\epsilon, p, (\mkObject{a_1 \mapsto \varnothing, \ldots, b_1 \mapsto t_1, \ldots}, \varnothing):e }
    \label{rule:tap-from-object}
    \\
    \mkState{\epsilon, .a:p, (\mkObject{a \mapsto u, \ldots}, o): e} &\longrightarrow \mkState{u, p, (\mkObject{a \mapsto u, \ldots}, o): e} \\
    \mkState{\epsilon, .a:p, (\mkObject{a \mapsto \varnothing, \ldots}, \{(a \mapsto (u, e')), \ldots\}): e} &\longrightarrow 
    \mkState{u, p, e'} \\
    \mkState{\epsilon, .a:p, (t, o): e} &\longrightarrow \mkState{\epsilon, .\varphi.a:p, (t, o): e} \text{,}\\ 
    \text{if $t \equiv \mkObject{\ldots}$, $\varphi \in \attr(t)$, $c \notin \attr(t)$} \span \notag\\
    \mkState{\epsilon, (a \mapsto (u, e')):p, (\mkObject{a \mapsto \varnothing, \ldots}, o): e} &\longrightarrow
    \mkState{\epsilon, p, (\mkObject{a \mapsto \varnothing, \ldots}, o \cup \{a \mapsto (u, e')\}): e} \text{,}\\
    \text{if $o$ is not defined for attribute $a$}\span\notag
\end{align}

  \end{mdframed}
  \caption{TAP machine for call-by-name evaluation of $\varphi$-terms.}
  \label{figure:tap-machine}
\end{figure*}

%% file: sections/5-translation.tex
\section{Translation to $\lambda$-calculus}
\label{section:translation-to-lambda-calculus}

In this section we compare $\varphi$-calculus with $\lambda$-calculus, present translation rules from one to the other, and prove soundness of the translation.

\subsection{$\lambda$-calculus with records}

We will use Mitchell Wand's $\lambda$-calculus with records \cite{WAND19911}, including both record extension (via $\text{\textbf{with}}$-expression) and record concatenation. As we will be translating locators approximately to de Bruijn indices in $\lambda$-terms, we will focus on a nameless variation of the syntax.

One important detail for us will be the computation rule for record extension. We will consider a term of the form $\recext{e}{\{\ldots\}}$ in weak head normal form, and extend the $\lambda$-calculus with the following evaluation rules:
\begin{align*}
    (\recext{e}{\{a = e_a, \ldots\}}).a &\longrightarrow e_a \\
    (\recext{e}{\{\ldots\}}).a &\longrightarrow e.a \quad \text{if $a$ is not in \{\ldots\}}
\end{align*}

This slight modification of Wand's calculus allows strictly more terms to avoid diverging computation, and is crucial for translation of decorators from $\varphi$-calculus.

\subsection{Translation from $\varphi$-calculus to $\lambda$-calculus}

\input{figures/translation}

To translate $\varphi$-terms to $\lambda$-terms, one must understand how to represent objects. Since records have nothing like void attributes, we cannot map void attributes directly to record attributes. So, instead, we will represent objects as functions taking records with instantiated void attributes. For example, we would like to represent an empty object $\mkObject{}$ as a constant function $\lambda\{\}$, and an object $\mkObject{x \mapsto \varnothing, y \mapsto \mkObject{}}$ as a function $\lambda\{x = 0.x, y = \lambda\{\}\}$.

Since locators enable referencing outer terms, for translation we will also make use of the fixpoint combinator. In particular, a term $\mkObject{x \mapsto \rho^0}$ should be translated into $\fix(\lambda\lambda\{x = \lambda (2\;(\recext{1}{0})\})$. Note that the outermost $\lambda$ introduces the translated object (represented as a function) that is then referenced as $2$ in $2\;(\recext{1}{0})$. $1$ references the instantiated void attributes passed to the outer term, and $0$ references the instantiated void attributes passed to the locator $\rho^0$.

In general, translation objects terms involves two $\lambda$-abstractions. One abstraction is used together with the fixed point combinator, to allow locators. And another one is used to properly represent void attributes. So, when translating locator $\rho^n$ we need to represent it so that it references the proper outer term and corresponding void attributes. That is why we translate $\rho^n$ to $\lambda \underline{2n+2} (\reccat{\underline{2n+1}}{\underline{0}})$. All the other translation rules follow naturally and are presented in Figure~\ref{fig:translation-rules}.

\subsection{Soundness of translation}

The translation is sound if it commutes with computation. That is, given $\varphi$-terms $t$ and $u$ such that $t \redto_\varphi u$, we have $\phitolambda(t) \approx \phitolambda(u)$. Intuitively, by $e_1 \approx e_2$ we mean that $e_1$ and $e_2$ are observationally equivalent. More precisely, $e_1 \approx e_2$ if and only if $e_1$ is $\beta\eta\zeta$-equivalent to $e_2$. Here by $\zeta$-equivalence we mean the obvious congruence rules, like associativity of $\|$: $\reccat{x}{(\reccat{y}{z})} \approx_\zeta \reccat{(\reccat{x}{y})}{z}$.

\begin{proposition}
    \label{proposition:phitolambda-preserves-redto}
    Let $t, u$ be $\varphi$-terms and $t \redto_\varphi u$. Then
    \[
        \phitolambda(t) \approx \phitolambda(u)
    \]
\end{proposition}
\begin{proof}
    Straightforward by induction on $t \redto_\varphi u$.
\end{proof}

\begin{theorem}[Soundness of $\phitolambda$]
    Let $t, t'$ be $\varphi$-terms 
    such that
    $t \redtomany_\varphi t'$. Then
    \[
        \phitolambda(t) \approx \phitolambda(t')
    \]
\end{theorem}
\begin{proof}
    Follows from Proposition~\ref{proposition:phitolambda-preserves-redto} and confluence of $\lambda$-calculus.
\end{proof}

\subsection{Translation from $\lambda$-calculus to $\varphi$-calculus}

The translation from $\varphi$-calculus aimed to map attributes of objects in $\varphi$-terms to attributes of records in $\lambda$-calculus. Unfortunately, such mapping is impossible in the backwards direction, unless we drop the record concatenation. Indeed, record concatenation cannot be translated to $\varphi$-calculus directly, as the latter does not support any mechanism for merging objects.

There exist two remaining options for translation from $\lambda$-calculus. First, we could translate only the segment without record concatenation. Such translation is possible under assumption that attributes map to attributes. However, as record concatenation is important for translation of locators from $\varphi$-calculus, this option is not ideal, as we only have full translation in one direction. Second, we can encode attributes and records, for example, using Church encoding. This would effectively translate $\lambda$-terms with records to mere $\lambda$-terms, which can then be translated to $\varphi$-calculus.

As we do not see a satisfactory translation from $\lambda$-calculus to $\varphi$-calculus, we propose, as a potential future work, an extension of $\varphi$-calculus with object concatenation.


\subsection{$\varphi$-calculus versus $\lambda$-calculus}

$\varphi$-calculus shares some common features with various $\lambda$-calculi as both are confluent term-rewriting systems capturing the notion of computability. Yet, $\varphi$-calculus focused on objects differs from $\lambda$-calculus in the following important ways:
\begin{enumerate}
  \item $\varphi$-calculus does not rely on $\lambda$-terms to represent functions. In $\varphi$-calculus everything is an object (in the sense of Definition~\ref{definition:object}).
  \item The attribute access in $\varphi$-calculus is more powerful than that of $\lambda$-calculus with records, since evaluation of the former involves substitution, essentially incorporating the expressive power of a fixpoint combinator. In fact attribute access shares certain similarities with with $\beta$-reduction in $\lambda$-calculus because of the substitution involved.
  \item Because of the decorators, $\varphi$-calculus requires no explicit analogue of fixpoint combinator. In a sense, objects have a recursive $\mathsf{let}$-construction built into them.
  \item The locators in $\varphi$-calculus are arguably more natural than de Bruijn indices used in $\lambda$-calculi. Here, by ``more natural'' we mean the following properties of locators as compared to de Bruijn indices:
  \begin{enumerate}
      \item In $\lambda$-calculi de Bruijn indices are used primarily for the convenience of algorithms, while humans prefer named function arguments. In $\varphi$-calculus, most objects already have a name as they are typically bound to some attribute. In fact, as we show in Section 6.1, locators can be omitted (most of the time).
      \item When presented in text (for humans), nested objects typically are indented, so for many examples it is easy to see which outer object the locator refers to. This is less convenient for $\lambda$-calculi where body of $\lambda$-abstraction is rarely indented. Besides, curried functions have arguments numbered in "reversed" order when using de Bruijn indices: $\lambda x_0. \lambda x_1. \lambda x_2. (x_0)\;((x_1)\;x_2)$ becomes $\lambda\lambda\lambda (2)\;((1)\;0)$ when using de Bruijn indices.
  \end{enumerate}
\end{enumerate}

Note that $\varphi$-calculus is Turing complete, as can be shown by embedding pure $\lambda$-calculus (in de Bruijn notation) into $\varphi$-calculus:
\begin{enumerate}
    \item (var) $n \longleftrightarrow \rho^n.\mathsf{arg}$
    \item (abs) $(t \longleftrightarrow u) \Rightarrow (\lambda t \;\longleftrightarrow\;\mkObject{ \mathsf{arg} \mapsto \varnothing, \mathsf{body} \mapsto u })$
    \item (app) $((t_1 \longleftrightarrow u_1) \land (t_2 \longleftrightarrow u_2)) \Rightarrow ((t_1)\;t_2 \longleftrightarrow u_1(\mathsf{arg} \mapsto u_2).\mathsf{body})$
\end{enumerate}

This embedding can be shown to be faithful in the sense that whenever $t \rightsquigarrow u$ we can encode $t$ in $\varphi$-calculus, compute, and decode the result: $t \longleftrightarrow t_\varphi \rightsquigarrow^* u_\varphi \longleftrightarrow u$.

As both calculi are Turing complete, it is not surprising that $\lambda$-calculus can be embedded into $\varphi$-calculus, or vice versa. However, $\varphi$-calculus shares enough similarities with $\lambda$-calculus with records to enable a particular kind of a translation, where objects of $\varphi$-calculus are, more or less, directly transated as records of $\lambda$-calculus. Such a direct translation is possible in one direction, and a partial\footnote{it is impossible to translate concatenation of records in this way} translation is possible the other. These translations can be used not only to improve understanding of the two formalizations, but also to translate certain useful properties between the systems. In particular, such a translation might be useful to develop a sound type system for $\varphi$-calculus in the future.

%% file: figures/translation.tex
\begin{figure*}
  \begin{mdframed}
    \emph{Syntax of nameless $\lambda$-calculus with records}
    \begin{align*}
        e &:= \underline{n} \tag{de Bruijn index} \\
          &\mid \lambda e
            \tag{abstraction} \\
          &\mid (e_1) e_2
            \tag{application} \\
          &\mid \{a_1 = e_1, \ldots, a_n = e_n\}
            \tag{record} \\
          &\mid e.a
            \tag{attribute} \\
          &\mid \recext{e}{\{a_1 = e_1, \ldots, a_n = e_n\}}
            \tag{record extension} \\
          &\mid e_1 \;\|\; e_2
            \tag{record concatenation} \\
          &\mid \fix e
            \tag{fixed point}
    \end{align*}
    
    \emph{Translation from $\varphi$-calculus to $\lambda$-calculus}
    \begin{align*}
        \phitolambda(\rho^n)
            &:= \lambda\underline{(2n + 2)} (\underline{(2n + 1)} \;\|\; \underline{0}) \\
        \phitolambda(t.a)
            &:= (\phitolambda(t)\;\{\}).a \\
        \phitolambda(t(a \mapsto u))
            &:= \lambda (\inc_\lambda(\phitolambda(t))) (\recext{\underline{0}}{\{a = \inc_\lambda(\phitolambda(u))\}}) \\
        \phitolambda(\mkObject{a_i \mapsto \varnothing, \ldots, b_j \mapsto t_j, \varphi \mapsto \ldots})
            &:= \fix (\lambda\lambda \recext{((\underline{1}\;\underline{0}).\varphi\;\{\}) }{\{a_i = 0.a_i, \ldots, b_j = \phitolambda(t_j) \}}) \\
        \phitolambda(\mkObject{a_i \mapsto \varnothing, \ldots, b_j \mapsto t_j})
            &:= \fix (\lambda\lambda \{\ldots, a_i = 0.a_i, \ldots, b_j = \phitolambda(t_j), \ldots\})
    \end{align*}
  \end{mdframed}

    \caption{Translation from $\varphi$-calculus to $\lambda$-calculus with records.}
    \label{fig:translation-rules}
\end{figure*}

%% file: sections/6-extensions.tex
\section{Extensions}

Bugayenko \cite{bugayenko/online} introduces a calculus that reflects capabilities of his EO programming language. As such it is richer than $\varphi$-calculus we have presented in Section~\ref{section:calculus}. In this section, we give examples of possible syntactic extensions to our calculus closing the gap between the two presentations. We leave out the more complicated extensions, such as mutable memory, primitive data types, or modelling input/output for future work.

\subsection{Attribute-variables}
\label{subsection:attribute-variables}

Locators are often used in combination with attribute access: $\rho^n.a$. In practice, though, attribute names can be descriptive and unique (at least in a certain scope or subterm) so that a person can understand easily which object this attribute belongs to. Such practice prompts a version of the syntax where locators are optional and can be omitted. For example, instead of $\mkObject{x \mapsto \rho^0.y, y \mapsto \mkObject{z \mapsto \rho^1.x}}$ one could omit both locators and it would still be clear where attributes should come from: $\mkObject{x \mapsto y, y \mapsto{z \mapsto x}}$.

More formally, we extend syntax of $\varphi$-terms with \emph{attribute-variables}:

\begin{definition}
  A set of $\varphi$-terms with attribute-variables $T_a$ is defined inductively as follows:
  \begin{enumerate}
    \item if $t \in T$ ($t$ is a $\varphi$-term) then $t \in T_a$;
    \item if $a \in \mathcal{L}$ then $a \in T_a$.
  \end{enumerate}
\end{definition}

As long as all attributes are defined in some enclosing object, we can restore locators. To do so we, have to traverse the term while keeping track of locators for known attributes. For the latter we will use a context represented by a mapping $\Gamma : \mathcal{L} \to \mathbb{N} \cup \{\bot\}$. For convenience we will write $\Gamma, a \in \rho^n$ to mean context $\Gamma'$ defined as follows:
\begin{align*}
  \Gamma'(a) &:= n \\
  \Gamma'(x) &:= \Gamma(x)\quad\text{when $x \not= a$}
\end{align*}

We will also define an increment operation on the context:
\begin{align*}
  \inct{}{\Gamma}(a) := \Gamma(a) + 1
\end{align*}

Translation from $\varphi$-terms with attribute-variables to regular $\varphi$-terms can be summarized with the following rules:
\begin{prooftree}
  \AxiomC{}
  \UnaryInfC{$\Gamma, a \in \rho^n \vdash a \longrightarrow \rho^n.a$}
  \DisplayProof\quad\quad
  \AxiomC{}
  \UnaryInfC{$\Gamma \vdash \rho^n \longrightarrow \rho^n$}
\end{prooftree}
\begin{prooftree}
  \AxiomC{$\Gamma \vdash t \longrightarrow t'$}
  \UnaryInfC{$\Gamma \vdash t.a \longrightarrow t'.a$}
  \DisplayProof\quad\quad
  \AxiomC{$\Gamma \vdash t \longrightarrow t'$}
  \AxiomC{$\Gamma \vdash u \longrightarrow u'$}
  \BinaryInfC{$\Gamma \vdash t(a \mapsto u) \longrightarrow t(a \mapsto u')$}
\end{prooftree}
\begin{prooftree}
  \AxiomC{$\inct{}{\Gamma}, a_0 \in \rho^0, \ldots, b_1 \in \rho^0, \ldots \vdash t_j \longrightarrow t'_j \quad \text{for all $j \in \{1, \ldots, n\}$}$}
  \UnaryInfC{$\Gamma \vdash \mkObject{a_1 \mapsto \varnothing, \ldots, b_1 \mapsto t_1, \ldots} \longrightarrow \mkObject{a_1 \mapsto \varnothing, \ldots, b_1 \mapsto t'_1, \ldots}$}
\end{prooftree}

Note that, by reversing the first rule, we can similarly erase unnecessary locators, yielding translation in the other direction.

\subsection{Global object}

Sometimes tracking nested objects might be inconvenient, and it might be easier to reference objects ``from the top-level''. This statement is especially true in an actual programming language. One can extend calculus with explicit names for objects to use instead of locators, but Section~\ref{subsection:attribute-variables} already provides a clean solution to provide a name for all terms, except for those at the top-level.

To reference top-level object by name, we may extend syntax with \emph{global object locator} $\Phi$. Similarly to attribute-variables, this extension is purely syntactic and requires no extension of evaluation rules as a proper locator can safely replace each occurrence of $\Phi$.

\subsection{Positional arguments}

Void attributes often serve as method arguments. To emphasize this role, we extend the syntax with positional arguments and nameless application.

We denote by
\[
\mkObject{\ldots, f(a_1, \ldots, a_k) \mapsto \mkObject{\ldots}, \ldots}
\]
an object where attribute $f$ is mapped to object with attributes $a_1, \ldots, a_k$ that are mapped to special void \emph{positional attributes} $\pi_1, \ldots, \pi_k$:
\[
\mkObject{\pi_1 \mapsto \varnothing, \ldots, \pi_k \mapsto \varnothing, a_1 \mapsto \rho^0.\pi_1, \ldots, a_k \mapsto \rho^0.\pi_k, \ldots}
\]

We denote by
\[ t\;t_1\;\ldots\;t_n \]
an application using positional attributes:
\[ t(\pi_1 \mapsto t_1, \ldots, \pi_n \mapsto t_n) \]

With this syntax abstract objects can be more easily identified as methods:

\begin{example}
    Using extended syntax, we can rewrite Example~\ref{example:phi-classes} as follows
\begin{align*}
    \mathsf{Base} &:= \mkObject{\mathsf{new} \mapsto  \\
                  &\phantom{:= \llbracket \mathsf{new}} \mkObject{ \mathsf{g}(\mathsf{this}) \mapsto \mkObject{\varphi \mapsto \mathsf{\rho^0.this.h(this \mapsto \rho^0.this).\varphi}},\\
                  &\phantom{:= \llbracket \mathsf{new} \llbracket} \mathsf{h}(\mathsf{this}) \mapsto \mkObject{\varphi \mapsto 3}}\\
                  &\phantom{:= \llbracket} }\\
 \mathsf{Derived} &:= \mkObject{\mathsf{new} \mapsto  \\
                  &\phantom{:= \llbracket \mathsf{new}} \mkObject{ \varphi \mapsto \mathsf{Base.new},\\
                  &\phantom{:= \llbracket \mathsf{new} \llbracket} \mathsf{f}(\mathsf{this}) \mapsto \mkObject{\varphi \mapsto \mathsf{2.add(n\mapsto\rho^0.this.g(this \mapsto \rho^0.this).\varphi)}},\\
                  &\phantom{:= \llbracket \mathsf{new} \llbracket} \mathsf{h}(\mathsf{this}) \mapsto \mkObject{\varphi \mapsto 5}}\\
                  &\phantom{:= \llbracket} } \\
d &:= \mathsf{Derived}.\mathsf{new}\span
\end{align*}
\end{example}





%% file: sections/7-related-work.tex
\section{Related work}

Lambda Calculus of Objects, $\lambda\mathrm{Obj}$~\cite{Fisher1993ALC} and its extensions (\cite{Ciaffaglione2021}) is perhaps the family of models that is closest to ours in spirit as they too deal with delegation-based inheritance. However, $\varphi$-calculus presents a somewhat minimal and pure (immutable) version without relying on $\lambda$-calculus. As we have seen, translation between $\varphi$-calculus and $\lambda$-calculus is not straightforward, so having smaller terms can be important in formal reasoning.

Systems based on row types and row polymorphism \cite{Wand1987CompleteTI} were originally introduced to model inheritance. Row types combine structural typing for records and variants with parametric polymorphism, which simplifies type inference. Rows can be extended (by adding new entries to the existing row) and concatenated (by combining several rows), which can be challenging for adoption in different typing settings and require new approaches \cite{Chlipala2010UrSM}. The last one introduces the Rose language, based on row types and supporting record concatenation through its monoidal nature of row extension. Rose uses qualified types to bind records to the rows and to abstract them from each other and allow them to evolve independently. It would be interesting to see whether row types can be used effectively to type $\varphi$-calculus.

%% file: sections/8-conclusion.tex
\section{Conclusion and future work}

In this paper, we have formalized $\varphi$-calculus, a calculus of objects with decoration as a primary mechanism of object extension. We have shown that even though our variant of $\varphi$-calculus is not based on $\lambda$-calculus, it possess the important properties, such as confluence (Church-Rosser property) and completeness of normal order evaluation.


Then we introduced an abstract machine for call-by-name evaluation of $\varphi$-terms. This machine can serve as reasoning tool for compilers and interpreters of $\varphi$-calculus and languages based on it, such as EO programming language.

We have also provided a sound translation from $\varphi$-calculus to $\lambda$-calculus with records. This translation emphasizes the differences between decoration and object extension using \textbf{with}-expression.
Finally, we discussed some syntactic extensions to the calculus, closing the gap between our presentation and that of Bugayenko \cite{bugayenko/online}.

We expect two main departures for the future work. First, we could add type system for the calculus, probably based on row types to facilitate type inference. We suspect that typed $\varphi$-calculus can be directly translated to $\lambda$-calculus with records and without concatenation operator, thus it would be possible to state equivalence (in a sense of having direct translation in both directions) between the two calculi. Second, we could extend the calculus with the ability to decorate or compose multiple objects, enabling simpler models for languages with multiple inheritance.

%% file: sections/A-complete-proofs.tex
\section{Complete proofs}

\subsection{Confluence}

\begin{proposition}[Reflexivity of parallel reduction]
    Let $t$ be a $\varphi$-term. Then $t \Rrightarrow t$.
\end{proposition}
\begin{proof}
    We prove this by induction on the structure of $t$:
    \begin{enumerate}
        \item if $t = u.c$ then by the inductive assumption $u \Rrightarrow u$ and by rule $\text{cong}_{\text{DOT}}^{\Rrightarrow}$ we have $u.c \Rrightarrow u.c$, i.e. $t \Rrightarrow t$;
        \item if $t = t_1(c \mapsto t_2)$ then by the inductive assumption $t_1 \Rrightarrow t_1$ and $t_2 \Rrightarrow t_2$; but then by $\text{cong}_{\text{APP}}^{\Rrightarrow}$ we have $t_1(c \mapsto t_2) \Rrightarrow t_1(c \mapsto t_2)$;
        \item if $t = \mkObject{a_1 \mapsto \varnothing, \ldots, a_k \mapsto \varnothing, b_1 \mapsto t_1, \ldots, b_n \mapsto t_n}$ then by inductive assumption we have $t_i \Rrightarrow t_i$ for each $i \in \{1, \ldots, n\}$ and by $\text{cong}_{\text{OBJ}}^{\Rrightarrow}$ we have $t \Rrightarrow t$;
        \item finally, if $t = \rho^n$ then $t \Rrightarrow t$ by $\text{cong}_{\rho}^{\Rrightarrow}$.
    \end{enumerate}
\end{proof}

\begin{definition}
    Relation $\redtomany$ is given by
    \begin{prooftree}
      \AxiomC{}
      \UnaryInfC{$t \redtomany t$}
      \DisplayProof\quad\quad
      \AxiomC{$t \redto t'$}
      \AxiomC{$t' \redtomany t''$}
      \BinaryInfC{$t \redtomany t''$}
    \end{prooftree}
\end{definition}

\begin{lemma}[Transitivity of $\redtomany$]
    \label{lemma:transitivity}
    For any $\varphi$-terms $t, t', t''$, if $t \redtomany t'$ and $t' \redtomany t''$, then  $t \redtomany t''$.
\end{lemma}

\begin{lemma}[Congruence reductions for $\redtomany$]
    For any $\varphi$-terms $t, t'$, $t \redtomany t'$ implies
    \begin{enumerate}
        \item $\mkObject{\ldots, b \mapsto t, \ldots} \redtomany \mkObject{\ldots, b \mapsto t', \ldots}$,
        \item $t.c \redtomany t'.c$,
        \item $t(c\mapsto u) \redtomany t'(c\mapsto u) $,
        \item $s(c\mapsto t)  \redtomany s(c\mapsto t') $.
    \end{enumerate}
\end{lemma}
\begin{proof}
     Proof by induction on the definition of $\redtomany$. 
     
     Assume as an induction hypothesis that $t \redtomany t'$ implies congruent reductions for $\redtomany$, stated above. Let $t'' \redto t$, so that $t'' \redtomany t$. Then 
        \begin{enumerate}
            \item $\mkObject{\ldots, b \mapsto t'', \ldots} \redto \mkObject{\ldots, b \mapsto t, \ldots}$ by $\text{cong}_{\text{OBJ}}$, so $\mkObject{\ldots, b \mapsto t'', \ldots} \redtomany \mkObject{\ldots, b \mapsto t', \ldots}$
            \item $t''.c \redto t.c$ by $\text{cong}_{\text{DOT}}$, so $t''.c \redtomany t'.c$
            \item $t''(c\mapsto u) \redto t(c\mapsto u)$ by cong$_{\text{APP}^\text{L}}$, so $t''(c\mapsto u) \redtomany t'(c\mapsto u)$
            \item $s.(c\mapsto t)  \redto s(c\mapsto t') $ by cong$_{\text{APP}^\text{R}}$, so $s(c\mapsto t'') \redtomany s(c\mapsto t')$
        \end{enumerate}
\end{proof}

\begin{proposition}[Equivalence of $\Rrightarrow$ and $\rightsquigarrow$]
    Parallel reduction ($\Rrightarrow$) is equivalent to regular reduction ($\rightsquigarrow$):
    \begin{enumerate}
        \item $t \rightsquigarrow t'$ implies $t \Rrightarrow t'$
        \item $t \redtomany t'$ implies $t \parredtomany t'$
        \item $t \Rrightarrow t'$ implies $t \redtomany t'$
        \item $t \parredtomany t'$ implies $t \redtomany t'$
    \end{enumerate}
\end{proposition}
\begin{proof}
    Since parallel reduction does ``zero or more'' reductions in a term, it is easy to see that regular reduction implies parallel reduction. On the other hand, a single step of parallel reduction implies several consecutive steps of regular reduction. This is a bit harder to prove, but is still rather straightforward.

    \begin{enumerate}
        \item If $t \rightsquigarrow t'$ by DOT$_{c}$, DOT$_{c}^\varphi$, or APP$_{c}$, then by reflexivity and corresponding rules of parallel reduction (DOT$_{c}^{\Rrightarrow}$, DOT$_{c}^{\varphi\Rrightarrow}$, APP$_{c}^{\Rrightarrow}$), $t \Rrightarrow t'$.
        To prove this implication for congruence reductions, assume as the induction hypothesis that $t \rightsquigarrow t'$ implies $t \Rrightarrow t'$. Then 
        \begin{enumerate}
            \item $\mkObject{\ldots, c \mapsto t, \ldots} \rightsquigarrow \mkObject{\ldots, c \mapsto t', \ldots}$ implies $\mkObject{\ldots, c \mapsto t, \ldots} \Rrightarrow \mkObject{\ldots, c \mapsto t', \ldots}$ by induction hypothesis, reflexivity of parallel reduction and $\text{cong}_{\text{OBJ}}^{\Rrightarrow}$.
            \item $t.c \rightsquigarrow t'.c$ implies $t.c \Rrightarrow t'.c$ by induction hypothesis and DOT$_{c}^{\Rrightarrow}$.
            \item $t(c \mapsto u) \rightsquigarrow t'(c \mapsto u)$ implies $t(c \mapsto u) \Rrightarrow t'(c \mapsto u)$ and $t_1(c \mapsto t) \rightsquigarrow t_1(c \mapsto t')$ implies  $t_1(c \mapsto t) \Rrightarrow t_1(c \mapsto t')$ by induction hypothesis, reflexivity of parallel reduction and APP$_{c}^{\Rrightarrow}$.
        \end{enumerate}
        \item $t \parredtomany t$ holds due to reflexivity; assume as an induction hypothesis that $t \redtomany t'$ implies $t \parredtomany t'$.
        As $t \redtomany t''$, these exists $t'$, such that $t \redto t'$ and $t' \redtomany t''$. By (1) of this proposition,  $t \Rrightarrow t'$, which, combined with the induction hypothesis, results in $t \parredtomany t''$.
        \item 
        Assume $t \Rrightarrow t'$ implies $t \redtomany t'$ as an induction hypothesis.
        \begin{enumerate}
        \item cong$_\text{OBJ}^{\Rrightarrow}$.
        If $t \equiv \mkObject{a_1 \mapsto \varnothing, \ldots, a_k \mapsto \varnothing, b_1 \mapsto t_1, \ldots, b_n \mapsto t_n}$ and $t' \equiv \mkObject{a_1 \mapsto \varnothing, \ldots, a_k \mapsto \varnothing, b_1 \mapsto t_1', \ldots, b_n \mapsto t_n'}$, then $t_i \Rrightarrow t_i'$ for $i \in \{1, \ldots, n\}$, and, by induction hypothesis, $t_i \redtomany t_i'$.
        
        Let $t^{\_i} \equiv \mkObject{a_1 \mapsto \varnothing, \ldots, a_k \mapsto \varnothing, b_1 \mapsto t_1', \ldots, b_i \mapsto t_i', b_{i+1} \mapsto t_{i+1}, \ldots b_n \mapsto t_n}$ for $i \in \{0, \ldots, n\}$.
        Observe that $t^{\_0} \equiv t$, $t^{\_n} \equiv t'$, and, by lemma for congruence reductions for ($\redtomany$), $t^{\_i-1} \redtomany t^{\_i}$. Finally, by lemma about transitivity, $t \redtomany t'$.
        
        \item cong$_\rho^{\Rrightarrow}$.
        $\rho^n \redtomany \rho^n$ holds due to reflexivity of ($\redtomany$).
        
        \item cong$_\text{DOT}^{\Rrightarrow}$.
        If $t.c \Rrightarrow t'.c$, then $t\Rrightarrow t'$, then $t\redtomany t'$ by induction hypothesis and $t.c\redtomany t'.c$ by lemma 2.
        
        \item cong$_\text{APP}^{\Rrightarrow}$.
        If $t(c\mapsto u) \Rrightarrow t'(c\mapsto u')$, then $t\Rrightarrow t'$ and $u\Rrightarrow u'$. By induction hypothesis, $t\redtomany t'$ and $u\redtomany u'$. By lemma 2, $t(c\mapsto u) \redtomany t'(c\mapsto u)$ and $t'(c\mapsto u) \redtomany t'(c\mapsto u')$, which can be combined by lemma 1 to yield $t(c\mapsto u) \redtomany t'(c\mapsto u')$.
        
        \item DOT$_{c}^{\Rrightarrow}$.
        Let reduction $t.c \Rrightarrow t_c\left[ \xi \mapsto t' \right]$ happen because $t \Rrightarrow t'$ and $t' \equiv \mkObject{\ldots, c \mapsto t_c, \ldots}$. By inductive hypothesis, $t \redtomany t'$; by lemma 2,  $t.c \redtomany t'.c$. By rule DOT$_{c}$ of regular reduction, $t'.c \redto t_c\left[ \xi \mapsto t' \right]$ (and hence $t'.c \redtomany t_c\left[ \xi \mapsto t' \right]$). By lemma about transitivity, $t.c \redtomany t_c\left[ \xi \mapsto t' \right]$.

        \item Proof for other rules (DOT$_{c}^{\varphi\Rrightarrow}$, APP$_{c}^{\Rrightarrow}$) is analogous to the one for DOT$_{c}^{\Rrightarrow}$: by induction hypothesis, establish $t \redtomany t'$ (and $u \redtomany u'$ in case of APP), apply congruence for rt-closure of regular reduction, use respective rule of regular reduction, and combine results with the lemma about transitivity.
        \end{enumerate}
        \item $t \redtomany t$ holds by definition. Assume as an induction hypothesis that $t \parredtomany t'$ implies $t \redtomany t'$.
        Let $t \parredtomany t''$ hold because $t \Rrightarrow t'$ and $t' \parredtomany t''$. By (3) of this proposition,  $t \redtomany t'$. By the induction hypothesis, $t' \redtomany t''$, and with \ref{lemma:transitivity}, $t \redtomany t''$.
    \end{enumerate}
\end{proof}

\begin{proposition}[Proposition~\ref{prop:reduction-to-cd}]
    \label{appendix:prop:reduction-to-cd}
    Let $t, t'$ be $\varphi$-terms and $t \Rrightarrow t'$. Then $t'~\Rrightarrow~t^{+}$.
\end{proposition}
\begin{proof}
    Assume as an induction hypothesis that if $t \Rrightarrow t'$ then $t' \Rrightarrow t^{+}$.
    \begin{enumerate}
        \item cong$_\text{OBJ}^{\Rrightarrow}$:
        If $\mkObject{a_1 \mapsto \varnothing, \ldots, a_k \mapsto \varnothing, b_1 \mapsto t_1, \ldots, b_n \mapsto t_n} \Rrightarrow \mkObject{a_1 \mapsto \varnothing, \ldots, a_k \mapsto \varnothing, b_1 \mapsto t_1', \ldots, b_n \mapsto t_n'}$, then $t_i\Rrightarrow t_i'$ for all $i\in \{1,\ldots,n\}$. By induction hypothesis, $t_i\Rrightarrow t_i^+$, hence $\mkObject{a_1 \mapsto \varnothing, \ldots, a_k \mapsto \varnothing, b_1 \mapsto t_1', \ldots, b_n \mapsto t_n'} \Rrightarrow \mkObject{a_1 \mapsto \varnothing, \ldots, a_k \mapsto \varnothing, b_1 \mapsto t_1^+, \ldots, b_n \mapsto t_n^+}$ by cong$_\text{OBJ}^{\Rrightarrow}$.
        \item cong$_\rho^{\Rrightarrow}$: $t\equiv\rho^i \Rrightarrow \rho^i \equiv t' \equiv \rho^i \Rrightarrow \rho^i \equiv t^+$.
        \item cong$_\text{DOT}^{\Rrightarrow}$:
        If $t.c \Rrightarrow t'.c$, then $t\Rrightarrow t'$ and by induction hypothesis, $t'\Rrightarrow t^+$. 
        \begin{align*}t'.c \Rrightarrow
        \begin{cases}
            t_c[\rho^0 \mapsto t^{+}], & \text{by DOT$_c^{\Rrightarrow}$, if $t^{+} \equiv \mkObject{\ldots, c \mapsto t_c, \ldots}$} \\
            t^{+}.\varphi.c, & \text{by DOT$_c^{\varphi\Rrightarrow}$, if $c \notin \attr(t^{+})$ and $\varphi \in \attr(t^{+})$} \\
            t^{+}.c & \text{by cong$_\text{DOT}^{\Rrightarrow}$, otherwise}
        \end{cases} 
        \end{align*}
        Hence, $t'.c\Rrightarrow (t.c)^+$.
        \item cong$_\text{APP}^{\Rrightarrow}$:
        If $t(c \mapsto u)\Rrightarrow t'(c \mapsto u')$, then $t\Rrightarrow t'$ and $u\Rrightarrow u'$. By induction hypothesis, $t'\Rrightarrow t^+$ and $u'\Rrightarrow u^+$. 
        \begin{align*}t'(c\mapsto u') \Rrightarrow
        \begin{cases}
            \mkObject{\ldots, c \mapsto \inct{}{u^+}, \ldots}, & \text{by APP$_c^{\Rrightarrow}$, if $t^{+} \equiv \mkObject{\ldots, a \mapsto \varnothing, \ldots}$} \\
            t^{+}(c \mapsto u^{+}) & \text{by cong$_\text{APP}^{\Rrightarrow}$, otherwise}
        \end{cases}
        \end{align*}
        Hence, $t'(c\mapsto u')\Rrightarrow (t(c\mapsto u))^+$.
        \item DOT$_c^{\Rrightarrow}$:
        If $t.c\Rrightarrow t_c[\xi \mapsto t']$, where $t\Rrightarrow t'\equiv \mkObject{\ldots, c\mapsto t_c, \ldots}$, then, by induction hypothesis, $t' \Rrightarrow t^+$, and as there is unique rule that allows reduction of an object (cong$_\text{OBJ}^{\Rrightarrow}$), $t^+\equiv \mkObject{\ldots, c\mapsto t_c', \ldots}$, with $t_c \Rrightarrow t_c'$.
        By substitution lemma, $t_c[\xi \mapsto t'] \Rrightarrow t_c'[\xi\mapsto t_c^+] \equiv (t.c)^+$.
        \item DOT$_c^{\varphi\Rrightarrow}$:
        Let $t.c\Rrightarrow t'.\varphi.c$, where $t\Rrightarrow t'\equiv \mkObject{\ldots}$, and $c\notin\attr(t')$ and $\varphi\in\attr(t')$. By induction hypothesis, $t' \Rrightarrow t^+$, and $t^+\equiv\mkObject{\ldots}$, s.t. as for $t'$, $c\notin\attr(t^+)$ and $\varphi\in\attr(t^+)$. Hence, $(t.c)^+ \equiv t^+.\varphi.c$. By cong$_\text{DOT}^{\Rrightarrow}$, $t'.\varphi \Rrightarrow t^+.\varphi$ and $t'.\varphi.c \Rrightarrow t^+.\varphi.c$.
        \item APP$_c^{\Rrightarrow}$:
        Let $t(c\mapsto u)\Rrightarrow \mkObject{a_1 \mapsto \varnothing, \ldots, a_k \mapsto \varnothing, c \mapsto \inct{}{u'}, b_1 \mapsto t_1, \ldots, b_n \mapsto t_n}$, where $t\Rrightarrow t'\equiv \mkObject{a_1 \mapsto \varnothing, \ldots, a_k \mapsto \varnothing, c \mapsto \varnothing, b_1 \mapsto t_1, \ldots, b_n \mapsto t_n}$ and $u\Rrightarrow u'$.
        By induction hypothesis, $u' \Rrightarrow u^+$ and $t' \Rrightarrow t^+\equiv \mkObject{a_1 \mapsto \varnothing, \ldots, a_k \mapsto \varnothing, c \mapsto \varnothing, b_1 \mapsto t_1', \ldots, b_n \mapsto t_n'}$. Observe that for all $i\in\{1,\ldots,n\}$, $t_i\Rrightarrow t_i'$. Note that it is not required that $t_i' \equiv t_i^+$.
        Complete development of $(t(c\mapsto u)) $ is $ \mkObject{a_1 \mapsto \varnothing, \ldots, a_k \mapsto \varnothing, c \mapsto \inct{}{u^+}, b_1 \mapsto t_1', \ldots, b_n \mapsto t_n'}$. By cong$_\text{OBJ}^{\Rrightarrow}$, $\mkObject{a_1 \mapsto \varnothing, \ldots, a_k \mapsto \varnothing, c \mapsto \inct{}{u'}, b_1 \mapsto t_1, \ldots, b_n \mapsto t_n} \Rrightarrow t(c\mapsto u)^+$
        
    \end{enumerate}
    \end{proof}

\begin{lemma}[Substitutions reordering]
    \label{lemma:swap-substitutions}
    For all $i, j \in \mathbb{N}$, $j \leq i$, $$t[\rho^j \mapsto u][\rho^i \mapsto v] \equiv t[\rho^{i+1} \mapsto \inct{}{v}][\rho^j \mapsto u[\rho^i \mapsto v]]$$
    This lemma encapsulates equivalence between multiple substitutions resulting from several DOT$_c$ reductions performed in a different order, respecting nesting relationship between objects.
\end{lemma}

\begin{proof}
    By induction on $t$:
    \begin{enumerate}
        \item $t \equiv \rho^k$
        \begin{align*}
            &\text{if $k<j$ then} 
            &\rho^k[\rho^j \mapsto u][\rho^i \mapsto v] \equiv \rho^k[\rho^i \mapsto v] \equiv \rho^k\\
            &{}
            &\rho^k[\rho^{i+1} \mapsto \inct{}{v}][\rho^j \mapsto u[\rho^i \mapsto v]] \equiv \rho^k[\rho^j \mapsto u[\rho^i \mapsto v]]\equiv \rho^k\\
            &\text{if $k\equiv j$ then} 
            &\rho^j[\rho^j \mapsto u][\rho^i \mapsto v] \equiv u[\rho^i \mapsto v]\\
            &{}
            &\rho^j[\rho^{i+1} \mapsto \inct{}{v}][\rho^j \mapsto u[\rho^i \mapsto v]] \equiv \rho^j[\rho^j \mapsto u[\rho^i \mapsto v]]\equiv u[\rho^i \mapsto v]\\
            &\text{if $j < k \leq i$ then} 
            &\rho^k[\rho^j \mapsto u][\rho^i \mapsto v] \equiv \rho^{k-1}[\rho^i \mapsto v] \equiv \rho^{k-1} \\
            &{}
            &\rho^k[\rho^{i+1} \mapsto \inct{}{v}][\rho^j \mapsto u[\rho^i \mapsto v]] \equiv \rho^k[\rho^j \mapsto u[\rho^i \mapsto v]]\equiv \rho^{k-1}\\
            &\text{if $k \equiv i+1$ then} 
            &\rho^{i+1}[\rho^j \mapsto u][\rho^i \mapsto v] \equiv \rho^i[\rho^i \mapsto v] \equiv v \\
            &{}
            &\rho^{i+1}[\rho^{i+1} \mapsto \inct{}{v}][\rho^j \mapsto u[\rho^i \mapsto v]] \equiv \inct{}{v}[\rho^j \mapsto u[\rho^i \mapsto v]]\equiv v\\
            &\text{if $k \geq i+2$ then} 
            &\rho^k[\rho^j \mapsto u][\rho^i \mapsto v] \equiv \rho^{k-1}[\rho^i \mapsto v] \equiv \rho^{k-2}\\
            &{}
            &\rho^k[\rho^{i+1} \mapsto \inct{}{v}][\rho^j \mapsto u[\rho^i \mapsto v]] \equiv \rho^{k-1}[\rho^j \mapsto u[\rho^i \mapsto v]]\equiv \rho^{k-2}
        \end{align*}
        \item $t \equiv s.a$
        \begin{multline*}
            s.a[\rho^j \mapsto u][\rho^i \mapsto v] \equiv 
            s[\rho^j \mapsto u][\rho^i \mapsto v].a \equiv \\
            s[\rho^{i+1} \mapsto \inct{}{v}][\rho^j \mapsto u[\rho^i \mapsto v]].a \equiv 
            s.a[\rho^{i+1} \mapsto \inct{}{v}][\rho^j \mapsto u[\rho^i \mapsto v]]
        \end{multline*}
        \item $t \equiv s_1(a\mapsto s_2)$
        \begin{multline*}
            s_1(a\mapsto s_2)[\rho^j \mapsto u][\rho^i \mapsto v] \equiv 
            s_1[\rho^j \mapsto u][\rho^i \mapsto v](a\mapsto s_2[\rho^j \mapsto u][\rho^i \mapsto v]) \equiv \\
            s_1[\rho^{i+1} \mapsto \inct{}{v}][\rho^j \mapsto u[\rho^i \mapsto v]](a\mapsto s_2[\rho^{i+1} \mapsto \inct{}{v}][\rho^j \mapsto u[\rho^i \mapsto v]])\equiv \\
            s_1(a\mapsto s_2)[\rho^{i+1} \mapsto \inct{}{v}][\rho^j \mapsto u[\rho^i \mapsto v]]
        \end{multline*}
        \item $t \equiv \mkObject{a_1 \mapsto \varnothing, \ldots, a_k \mapsto \varnothing, b_1 \mapsto t_1, \ldots, b_n \mapsto t_n}$
        \begin{align*}
        t[\rho^j \mapsto u]&[\rho^i \mapsto v] \\
        &\equiv \mkObject{\ldots, b_l \mapsto t_l[\rho^{j+1} \mapsto u][\rho^{i+1} \mapsto v] , \ldots}\\
        &\equiv \mkObject{\ldots, b_l \mapsto t_l[\rho^{i+2} \mapsto \inct{}{v}][\rho^{j+1} \mapsto u[\rho^{i+1} \mapsto v]], \ldots}\\
        &\equiv t[\rho^{i+1} \mapsto \inct{}{v}][\rho^j \mapsto u[\rho^i \mapsto v]]
        \end{align*}
    \end{enumerate}
\end{proof}

\begin{lemma}[Increment and substitution swap]
    For any $\varphi$-term $t$, for any $i, j \in \mathbb{N}$ such that $j\leq i$, $\inct{j}{t}[\rho^{i+1}\mapsto \inct{j}{u}] \equiv \inct{j}{t[\rho^i\mapsto u]}$.
\end{lemma}
\begin{proof}
    By induction on $t$:
    \begin{enumerate}
        \item $t \equiv \rho^k$
        \begin{align*}
            &\text{if $k < j$ then} 
            &\inct{j}{\rho^k}[\rho^{i+1}\mapsto \inct{j}{u}] \equiv \rho^{k}[\rho^{i+1}\mapsto \inct{j}{u}] \equiv \rho^{k}\\
            &{}
            &\inct{j}{\rho^k[\rho^i\mapsto u]} \equiv\inct{j}{\rho^k}\equiv\rho^{k}\\
            &\text{if $j \leq k < i$ then} 
            &\inct{j}{\rho^k}[\rho^{i+1}\mapsto \inct{j}{u}] \equiv \rho^{k+1}[\rho^{i+1}\mapsto \inct{j}{u}] \equiv \rho^{k+1}\\
            &{}
            &\inct{j}{\rho^k[\rho^i\mapsto u]} \equiv\inct{j}{\rho^k}\equiv\rho^{k+1}\\
            &\text{if $k\equiv i$ then} 
            &\inct{j}{\rho^i}[\rho^{i+1}\mapsto \inct{j}{u}] \equiv \rho^{i+1}[\rho^{i+1}\mapsto \inct{j}{u}] \equiv \inct{j}{u}\\
            &{}
            &\inct{j}{\rho^i[\rho^i\mapsto u]} \equiv\inct{j}{u}\\
            &\text{if $k > i$ then} 
            &\inct{j}{\rho^k}[\rho^{i+1}\mapsto \inct{j}{u}] \equiv \rho^{k+1}[\rho^{i+1}\mapsto \inct{j}{u}] \equiv \rho^{k}\\
            &{}
            &\inct{j}{\rho^k[\rho^i\mapsto u]} \equiv\inct{j}{\rho^{k-1}}\equiv\rho^{k}\\
        \end{align*}
        \item $t \equiv s.a$
        $$\inct{j}{s.a}[\rho^{i+1}\mapsto \inct{j}{u}] \equiv \inct{j}{s}[\rho^{i+1}\mapsto \inct{j}{u}].a \equiv\inct{j}{s[\rho^i\mapsto u]}.a\equiv\inct{j}{s.a[\rho^i\mapsto u]}$$
        \item $t \equiv s_1(a\mapsto s_2)$
        \begin{align*}
            \inct{j}{ s_1(a\mapsto s_2)}[\rho^{i+1}\mapsto \inct{j}{u}] \equiv \inct{j}{s_1}[\rho^{i+1}\mapsto \inct{j}{u}](a\mapsto\inct{j}{s_2}[\rho^{i+1}\mapsto \inct{j}{u}])\\ \equiv\inct{j}{s_1[\rho^i\mapsto u]}(a\mapsto\inct{j}{s_2[\rho^i\mapsto u]}) \equiv\inct{j}{s_1(a\mapsto s_2)[\rho^i\mapsto u]}
        \end{align*}
        \item $t \equiv \mkObject{a_1 \mapsto \varnothing, \ldots, a_k \mapsto \varnothing, b_1 \mapsto t_1, \ldots, b_n \mapsto t_n}$. Proof by unfolding the definitions of substitution and increment, swapping  increments, applying induction hypothesis, and folding the definitions back:
        \begin{align*}
        \inct{j}{\mkObject{a_1 \mapsto \varnothing, \ldots, a_k \mapsto \varnothing, b_1 \mapsto t_1, \ldots, b_n \mapsto t_n}}[\rho^{i+1}\mapsto \inct{j}{u}] \equiv \\
        \mkObject{a_1 \mapsto \varnothing, \ldots, a_k \mapsto \varnothing, b_1 \mapsto \inct{j+1}{t_1}[\rho^{i+2}\mapsto \inct{}{\inct{j}{u}}], \ldots, b_n \mapsto \inct{j+1}{t_n}[\rho^{i+2}\mapsto \inct{}{\inct{j}{u}}]} \equiv \\
        \mkObject{a_1 \mapsto \varnothing, \ldots, a_k \mapsto \varnothing, b_1 \mapsto \inct{j+1}{t_1}[\rho^{i+2}\mapsto \inct{j+1}{\inct{}{u}}], \ldots, b_n \mapsto \inct{j+1}{t_n}[\rho^{i+2}\mapsto \inct{j+1}{\inct{}{u}}]}\equiv\\
        \mkObject{a_1 \mapsto \varnothing, \ldots, a_k \mapsto \varnothing, b_1 \mapsto \inct{j+1}{t_1[\rho^{i+1}\mapsto \inct{}{u}]}, \ldots, b_n \mapsto \inct{j+1}{t_n[\rho^{i+1}\mapsto \inct{}{u}]}}\equiv\\
        \inct{j}{\mkObject{a_1 \mapsto \varnothing, \ldots, a_k \mapsto \varnothing, b_1 \mapsto t_1, \ldots, b_n \mapsto t_n}[\rho^i\mapsto u]}
        \end{align*}
    \end{enumerate}
\end{proof}

\begin{lemma}[Increment swap]
    For any $\varphi$-term $t$, for any $i, j \in \mathbb{N}$ such that $i\leq j$, $\inct{i}{\inct{j}{t}} \equiv \inct{j+1}{\inct{i}{t}}$.
\end{lemma}
\begin{proof}
    By induction on $t$:
    \begin{enumerate}
        \item $t \equiv \rho^k$
        \begin{align*}
            &\text{if $k < i$ then} 
            &\inct{i}{\inct{j}{\rho^k}} \equiv \inct{i}{\rho^k} \equiv \rho^k\\
            &{}
            & \inct{j+1}{\inct{i}{\rho^k}}\equiv \inct{j+1}{\rho^k} \equiv \rho^k\\
            &\text{if $i \leq k < j$ then} 
            &\inct{i}{\inct{j}{\rho^k}} \equiv \inct{i}{\rho^k} \equiv \rho^{k+1}\\
            &{}
            & \inct{j+1}{\inct{i}{\rho^k}}\equiv \inct{j+1}{\rho^{k+1}} \equiv \rho^{k+1}\\
            &\text{if $k\geq j$ then} 
            &\inct{i}{\inct{j}{\rho^k}} \equiv \inct{i}{\rho^{k+1}} \equiv \rho^{k+2}\\
            &{}
            & \inct{j+1}{\inct{i}{\rho^k}}\equiv \inct{j+1}{\rho^{k+1}} \equiv \rho^{k+2}\\
            \end{align*}
        \item $t \equiv s.a$
        \item $t \equiv s_1(a\mapsto s_2)$
        \item $t \equiv \mkObject{a_1 \mapsto \varnothing, \ldots, a_k \mapsto \varnothing, b_1 \mapsto t_1, \ldots, b_n \mapsto t_n}$
        \begin{align*}
        \inct{i}{\inct{j}{\mkObject{a_1 \mapsto \varnothing, \ldots, a_k \mapsto \varnothing, b_1 \mapsto t_1, \ldots, b_n \mapsto t_n}}} \equiv\\
        \mkObject{a_1 \mapsto \varnothing, \ldots, a_k \mapsto \varnothing, b_1 \mapsto \inct{i+1}{\inct{j+1}{t_1}}, \ldots, b_n \mapsto \inct{i+1}{\inct{j+1}{t_n}}} \equiv\\
        \mkObject{a_1 \mapsto \varnothing, \ldots, a_k \mapsto \varnothing, b_1 \mapsto \inct{j+2}{\inct{i+1}{t_1}}, \ldots, b_n \mapsto \inct{j+2}{\inct{i+1}{t_n}}} \equiv\\
        \inct{j+1}{\inct{i}{\mkObject{a_1 \mapsto \varnothing, \ldots, a_k \mapsto \varnothing, b_1 \mapsto t_1, \ldots, b_n \mapsto t_n}}}
        \end{align*}
    \end{enumerate}
\end{proof}

\begin{lemma}[Substitution lemma]
    Let $t, t', u, u'$ be $\varphi$-terms and $t \Rrightarrow t'$ and $u \Rrightarrow u'$.
    Then $t[\rho^i \mapsto u] \Rrightarrow t'[\rho^i \mapsto u']$.
\end{lemma}
\begin{proof}
    To prove by induction on $\Rrightarrow$, assume that if $t \Rrightarrow t'$ and $u \Rrightarrow u'$, then for all $i\in\mathbb{N}$, $t[\rho^i \mapsto u] \Rrightarrow t'[\rho^i \mapsto u']$.
    \begin{enumerate}
    \item
    cong$_{\text{OBJ}}^{\Rrightarrow}$:
    
    By cong$_{\text{OBJ}}^{\Rrightarrow}$ and induction hypothesis,
    \begin{multline*}
        \mkObject{a_1 \mapsto \varnothing, \ldots, a_k \mapsto \varnothing, b_1 \mapsto t_1, \ldots, b_n \mapsto t_n}[\rho^k \mapsto u] \\
        \equiv \mkObject{a_1 \mapsto \varnothing, \ldots, a_k \mapsto \varnothing, b_1 \mapsto t_1[\rho^{k+1} \mapsto u], \ldots, b_n \mapsto t_n[\rho^{k+1} \mapsto u]} \\
        \Rrightarrow \mkObject{a_1 \mapsto \varnothing, \ldots, a_k \mapsto \varnothing, b_1 \mapsto t_1'[\rho^{k+1} \mapsto u'], \ldots, b_n \mapsto t_n'[\rho^{k+1} \mapsto u']} \\
        \equiv\mkObject{a_1 \mapsto \varnothing, \ldots, a_k \mapsto \varnothing, b_1 \mapsto t_1', \ldots, b_n \mapsto t_n'}[\rho^k \mapsto u']
    \end{multline*}
    \item
    cong$_\rho^{\Rrightarrow}$:
    \begin{itemize}
        \item if $k>i$, $\rho^i[\rho^k \mapsto u] \equiv \rho^i \Rrightarrow \rho^i \equiv \rho^i[\rho^k \mapsto u']$, with reflexivity of ($\Rrightarrow$).
        \item if $k=i$, $\rho^i[\rho^k \mapsto u] \equiv u \Rrightarrow u' \equiv \rho^i[\rho^k \mapsto u']$, with the assumption of this lemma.
        \item if $k<i$, $\rho^i[\rho^k \mapsto u] \equiv \rho^{i-1} \Rrightarrow \rho^{i-1} \equiv \rho^i[\rho^k \mapsto u']$, with reflexivity of ($\Rrightarrow$).
    \end{itemize}
    \item
    cong$_\text{DOT}^{\Rrightarrow}$:
    
    By inductive hypothesis and cong$_\text{DOT}^{\Rrightarrow}$,
    \begin{align*}
        t.c [\rho^i \mapsto u] \equiv 
        t [\rho^i \mapsto u].c \Rrightarrow
        t' [\rho^i \mapsto u'].c \equiv
        & t'.c [\rho^i \mapsto u'] 
    \end{align*}
    
    \item
    cong$_\text{APP}^{\Rrightarrow}$:
    
    By inductive hypothesis and cong$_\text{APP}^{\Rrightarrow}$,
    \begin{align*}
        t(c\mapsto v) [\rho^i \mapsto u] \equiv 
        t [\rho^i \mapsto u] (c \mapsto v [\rho^i \mapsto u])\Rrightarrow
        t' [\rho^i \mapsto u'] (c \mapsto v' [\rho^i \mapsto u']) \equiv
        t'(c\mapsto v')  [\rho^i \mapsto u'] 
    \end{align*}
    
    \item
    DOT$_{c}^{\Rrightarrow}$:
    Let $t \equiv s.c$, $s \Rrightarrow s' \equiv \mkObject{\ldots, c \mapsto t_c, \ldots}$, $t' \equiv t_c[\xi \mapsto s']$.
    
    Since $s \Rrightarrow s'$ and  $u \Rrightarrow u'$, by induction hypothesis we have $$ s [\rho^i\mapsto u] \Rrightarrow  s'[\rho^i\mapsto u']$$
    
    Furthermore, since $s' \equiv \mkObject{\ldots, c \mapsto t_c, \ldots}$, $$s'[\rho^i\mapsto u'] \equiv \mkObject{\ldots, c \mapsto t_c[\rho^{i+1}\mapsto u'] , \ldots}$$
    and by rule DOT$_c^{\Rrightarrow}$ we have
    $$s [\rho^i\mapsto u].c \Rrightarrow t_c[\rho^{i+1}\mapsto u'] [\xi \mapsto s'[\rho^i\mapsto u']]$$
    By \ref{lemma:swap-substitutions}, we have
    $$t_c[\rho^{i+1}\mapsto u'] [\xi \mapsto s'[\rho^i\mapsto u']] \equiv t_c[\xi \mapsto s'][\rho^i\mapsto u']$$
    Thus
    \begin{multline*}
        t[\rho^i\mapsto u] \equiv s.c [\rho^i\mapsto u] \equiv s [\rho^i\mapsto u].c \Rrightarrow \\
        t_c[\rho^{i+1}\mapsto u'] [\xi \mapsto s'[\rho^i\mapsto u']] \equiv 
        t_c[\xi \mapsto s'][\rho^i\mapsto u'] \equiv t'[\rho^i\mapsto u']
    \end{multline*}
    \item
    DOT$_{c}^{\varphi\Rrightarrow}$:
    Let $t \equiv s.c$ and $t' \equiv s'.\varphi.c$ with $s\Rrightarrow s'$.
    
    By induction hypothesis, $s[\rho^i\mapsto u] \Rrightarrow s'[\rho^i\mapsto u']$. Substitution in object does not change its attributes set (domain): $c\notin\attr(s'[\rho^i\mapsto u'])$ and $\varphi \in \attr(s'[\rho^i\mapsto u'])$. So, rule DOT$_{c}^{\varphi\Rrightarrow}$ can be applied to get
    $$s[\rho^i\mapsto u] .c \Rrightarrow s'[\rho^i\mapsto u'].\varphi.c$$ which, with the definition of substitution, brings to the goal:
    $$s.c[\rho^i\mapsto u]  \Rrightarrow s'.\varphi.c[\rho^i\mapsto u']$$
    
    \item
    APP$_{c}^{\Rrightarrow}$:
    Let $t \equiv s(c\mapsto v)$, $v \Rrightarrow v'$, $s' \equiv \mkObject{a_1 \mapsto \varnothing, \ldots, a_k \mapsto \varnothing, c \mapsto \varnothing, b_1 \mapsto t_1, \ldots, b_n \mapsto t_n}$ and $t' \equiv \mkObject{a_1 \mapsto \varnothing, \ldots, a_k \mapsto \varnothing, c \mapsto \inct{}{v'}, b_1 \mapsto t_1, \ldots, b_n \mapsto t_n}$.
    
    By induction hypothesis, $s[\rho^i\mapsto u] \Rrightarrow s'[\rho^i\mapsto u']$ and $v[\rho^i\mapsto u] \Rrightarrow v'[\rho^i\mapsto u']$. As in $s'$, in $s'[\rho^i\mapsto u']$ there is attribute $c$ and it is free. So, reduction via APP$_{c}^{\Rrightarrow}$ is possible.
    
    \begin{multline*}
        t[\rho^i\mapsto u] 
        \equiv s(c\mapsto v)[\rho^i\mapsto u]
        \equiv s[\rho^i\mapsto u](c\mapsto v[\rho^i\mapsto u])\\
        \Rrightarrow  \mkObject{a_1 \mapsto \varnothing, \ldots, a_k \mapsto \varnothing, c \mapsto \inct{}{v'[\rho^{i}\mapsto u']}, b_1 \mapsto t_1[\rho^{i+1}\mapsto \inct{}{u'}], \ldots, b_n \mapsto t_n[\rho^{i+1}\mapsto \inct{}{u'}]} \\
        \equiv  \mkObject{a_1 \mapsto \varnothing, \ldots, a_k \mapsto \varnothing, c \mapsto \inct{}{v'}[\rho^{i+1}\mapsto \inct{}{u'}], b_1 \mapsto t_1[\rho^{i+1}\mapsto \inct{}{u'}], \ldots, b_n \mapsto t_n[\rho^{i+1}\mapsto \inct{}{u'}]} \\
        \equiv  \mkObject{a_1 \mapsto \varnothing, \ldots, a_k \mapsto \varnothing, c \mapsto \inct{}{v'}, b_1 \mapsto t_1, \ldots, b_n \mapsto t_n}[\rho^i\mapsto u']
    \end{multline*}
    \end{enumerate}
\end{proof}

\begin{proposition}
    Let $t$ be a $\varphi$-term. Then $t \Rrightarrow t^{+}$.
\end{proposition}
\begin{proof}
    By induction on structure of $t$,
    \begin{enumerate}
        \item if $t\equiv\rho^n$, then by cong$_\rho^{\Rrightarrow}$, $t \Rrightarrow \rho^n \equiv (\rho^n)^+$.
        \item if $t\equiv\mkObject{a_1 \mapsto \varnothing, \ldots, a_k \mapsto \varnothing, b_1 \mapsto t_1, \ldots, b_n \mapsto t_n}$, then by induction hypothesis, $t_i \Rrightarrow t_i^+$, and by cong$_\text{OBJ}^{\Rrightarrow}$, $t \Rrightarrow \mkObject{a_1 \mapsto \varnothing, \ldots, a_k \mapsto \varnothing, b_1 \mapsto t_1^+, \ldots, b_n \mapsto t_n^+} \equiv t^+$.
        \item if $t\equiv t_1.c$, then by induction hypothesis, $t_1\Rrightarrow t_1^+$, and
        \begin{enumerate}
            \item if $t_1^+ \equiv \mkObject{\ldots, a\mapsto t_a, \ldots}$, then by DOT$_{c}^{\Rrightarrow}$, $t\Rrightarrow t_a[\phi^0 \mapsto t_1^+]$
            \item if $c \notin \attr(t_1^{+})$ and $\varphi \in \attr(t_1^{+})$, then by DOT$_{c}^{\varphi\Rrightarrow}$, $t\Rrightarrow t_1^{+}.\varphi.c \equiv t^+$
            \item otherwise, by cong$_{DOT}^{\Rrightarrow}$, $t\Rrightarrow t_1^{+}.c \equiv t^+$
        \end{enumerate}
        \item if $t\equiv t_1(c\mapsto u)$, then by induction hypothesis, $t_1\Rrightarrow t_1^+$, and $u \Rrightarrow u^+$, and
        \begin{enumerate}
            \item if $t^{+} \equiv \mkObject{\ldots, \varphi \mapsto t_\varphi, \ldots}$, then by APP$_{c}^{\Rrightarrow}$, $t\Rrightarrow \mkObject{\ldots, c \mapsto \inct{}{u^{+}}, \ldots}$
            \item otherwise, by cong$_{APP}^{\Rrightarrow}$, $t\Rrightarrow t_1^{+}(c\mapsto u^+) \equiv t^+$
        \end{enumerate}
    \end{enumerate}
\end{proof}

\begin{lemma}[Main Lemma, Lemma~\ref{lemma:normal-reduction:main}]
    $t \Rrightarrow s$ implies $t \headredmany r \innerparred s$ for some $r$.
\end{lemma}

\begin{proof}

By induction on $\Rrightarrow$,
\begin{enumerate}
    \item cong$_\text{OBJ}^{\Rrightarrow}$:
        If $t\equiv \mkObject{a_1 \mapsto \varnothing, \ldots, a_k \mapsto \varnothing, b_1 \mapsto t_1, \ldots, b_n \mapsto t_n} \Rrightarrow \mkObject{a_1 \mapsto \varnothing, \ldots, a_k \mapsto \varnothing, b_1 \mapsto t_1', \ldots, b_n \mapsto t_n'} \equiv s$, where $t_i\Rrightarrow t_i'$ for all $i\in \{1,\ldots,n\}$, then let $r \equiv t$, and from the definition of $\innerparred$, $t \headredmany r \innerparred s$.
        \item cong$_\rho^{\Rrightarrow}$: 
        If $t\equiv\rho^i \Rrightarrow \rho^i \equiv s$, let $r \equiv t$.
        \item cong$_\text{DOT}^{\Rrightarrow}$:
        If $t\equiv t'.c \Rrightarrow s'.c \equiv s$ where $t'\Rrightarrow s'$, then by induction hypothesis, there exists $r'$, such that $t' \headredmany r' \innerparred s'$. Hence, $t'.c \headredmany r'.c \innerparred s'.c$, from the definition of $\innerparred$ and $\headred$. 
        \item cong$_\text{APP}^{\Rrightarrow}$:
        If $t\equiv t'(c \mapsto u) \Rrightarrow s'(c \mapsto u') \equiv s$, where $t' \Rrightarrow s'$ and $u\Rrightarrow u'$, then by induction hypothesis, there exists $r'$, such that $t' \headredmany r' \innerparred s'$. Hence, $t'(c\mapsto u) \headredmany r'(c\mapsto u) \innerparred s'(c\mapsto u')$, from the definition of $\innerparred$ and $\headred$. 
        \item DOT$_c^{\Rrightarrow}$:
        If $t\equiv t'.c\Rrightarrow t_c[\rho^0 \mapsto t'']\equiv s$, where $t'\Rrightarrow t''\equiv \mkObject{\ldots, c\mapsto t_c, \ldots}$, then, by induction hypothesis, there exists $q$, such that $t' \headredmany q \innerparred t''$. Since $q \innerparred t''$ and $t''\equiv \mkObject{\ldots, c\mapsto t_c, \ldots}$, $q \equiv \mkObject{\ldots, c\mapsto t_c', \ldots}$ with $t_c' \Rrightarrow t_c$. Then $q.c \headred t_c'[\rho^0 \mapsto q]$.
        Moreover, we have $t_c' \Rrightarrow t_c$, and by induction hypothesis, there exists $r'$, such that $t_c' \headredmany r' \innerparred t_c$.
        Finally, with substitution lemmas for $\innerparred$ and $\headred$, we have $t\equiv t'.c \headredmany q.c \headred t_c'[\rho^0 \mapsto q] \headredmany r'[\rho^0 \mapsto q] \innerparred t_c[\rho^0 \mapsto t''] \equiv s$, so we can take $r'[\rho^0 \mapsto q]$ for  $r$.
        \item DOT$_c^{\varphi\Rrightarrow}$:
        If $t\equiv t'.c\Rrightarrow s'.\varphi.c \equiv s$, where $t'\Rrightarrow s'\equiv \mkObject{\ldots}$, and $c\notin\attr(s')$ and $\varphi\in\attr(s')$, then, by induction hypothesis, there exists $r'$, such that $t'\headredmany r' \innerparred s'$. As $r' \innerparred s'$, $r' \equiv \mkObject{\ldots}$ with $c\notin\attr(r')$ and $\varphi\in\attr(r')$. Hence $r'.c \headred r'.\varphi.c$. Secondly, $r'.\varphi.c \innerparred s'.\varphi.c$. So, for $r \equiv r'.\varphi.c$, $t\headredmany r \innerparred s$.
        \item APP$_c^{\Rrightarrow}$:
        If $t\equiv t'(c\mapsto u)\Rrightarrow \mkObject{a_1 \mapsto \varnothing, \ldots, a_k \mapsto \varnothing, c \mapsto \inct{}{u'}, b_1 \mapsto t_1, \ldots, b_n \mapsto t_n} \equiv s$, where $t'\Rrightarrow s'\equiv \mkObject{a_1 \mapsto \varnothing, \ldots, a_k \mapsto \varnothing, c \mapsto \varnothing, b_1 \mapsto t_1, \ldots, b_n \mapsto t_n}$ and $u\Rrightarrow u'$, then, by induction hypothesis, there exists $r'$, such that $t' \headredmany r' \innerparred s'$. As $r' \innerparred s'$, $r' \equiv \mkObject{a_1 \mapsto \varnothing, \ldots, a_k \mapsto \varnothing, c \mapsto \varnothing, b_1 \mapsto t_1', \ldots, b_n \mapsto t_n'}$ with $t_i' \Rrightarrow t_i$. Hence $r'(c\mapsto u) \headred \mkObject{a_1 \mapsto \varnothing, \ldots, a_k \mapsto \varnothing, c \mapsto \inct{}{u}, b_1 \mapsto t_1', \ldots, b_n \mapsto t_n'}$. Let $r \equiv \mkObject{a_1 \mapsto \varnothing, \ldots, a_k \mapsto \varnothing, c \mapsto \inct{}{u}, b_1 \mapsto t_1', \ldots, b_n \mapsto t_n'}$. By cong$_{\text{OBJ}}^{\innerparred}$, $r \innerparred s$, so, $t\headredmany r \innerparred s$.
\end{enumerate}

\end{proof}

\begin{lemma}[Substitution Lemma for $\headred$, Lemma~\ref{{lemma:normal-reduction:substitution}}]
    If $t \headred s$, then $t[\rho^n \mapsto q] \headred s[\rho^n \mapsto q]$.
\end{lemma}

\begin{proof}
    By induction on $\headred$:
    \begin{enumerate}
        \item cong$_{\text{DOT}}^h$:
        If $t.a \headred t'.a$ and $t \headred t'$, then by induction hypothesis $t[\rho^n \mapsto q] \headred t'[\rho^n \mapsto q]$, so $t.a[\rho^n \mapsto q] \headred t'.a[\rho^n \mapsto q]$
        \item cong$_{\text{APP}}^h$:
        If $t(a \mapsto u) \headred t'(a \mapsto u)$ and $t \headred t'$, then by induction hypothesis $t[\rho^n \mapsto q] \headred t'[\rho^n \mapsto q]$, so $t[\rho^n \mapsto q](a \mapsto u[\rho^n \mapsto q]) \headred t'[\rho^n \mapsto q](a\mapsto u[\rho^n \mapsto q])$, and $t(a \mapsto u)[\rho^n \mapsto q] \headred t'(a \mapsto u)[\rho^n \mapsto q]$.
        \item DOT$_{c}$:
        If $t.c \headred t_c \left[ \rho^0 \mapsto t \right]$ and $t \equiv \mkObject{\ldots, \textcolor{black}{c \mapsto t_c}, \ldots}$, then $t[\rho^n \mapsto q] \equiv \mkObject{\ldots, \textcolor{black}{c \mapsto t_c[\rho^{n+1} \mapsto \inct{}{q}]}, \ldots}$, and $t.c[\rho^n \mapsto q] \headred t_c[\rho^{n+1} \mapsto \inct{}{q}][\rho^0 \mapsto t[\rho^n \mapsto q]] \equiv t_c \left[ \rho^0 \mapsto t \right][\rho^n \mapsto q]$ by the Reordering Substitutions lemma.
        \item DOT$_{c}^\varphi$:
        If $t.c \headred t.\varphi.c$ and $c \notin \attr(t)$ $\varphi \in \attr(t)$ $t \equiv \mkObject{\ldots}$, then $t[\rho^n \mapsto q].c \headred t[\rho^n \mapsto q].\varphi.c$, as substitution does not change set of attributes. So, $t.c[\rho^n \mapsto q] \headred t.\varphi.c[\rho^n \mapsto q]$.
        \item APP$_{c}$:
        If $t (c \mapsto u) \headred \mkObject{a_1 \mapsto \varnothing, \ldots, a_k \mapsto \varnothing, \textcolor{black}{c \mapsto \inct{}{u}}, b_1 \mapsto t_1, \ldots, b_n \mapsto t_n}$ and $t \equiv \mkObject{a_1 \mapsto \varnothing, \ldots, a_k \mapsto \varnothing, \textcolor{black}{c \mapsto \varnothing}, b_1 \mapsto t_1, \ldots, b_n \mapsto t_n}$, then $t (c \mapsto u) [\rho^n \mapsto q] \equiv t[\rho^n \mapsto q] (c \mapsto u[\rho^n \mapsto q]) \headred \mkObject{a_1 \mapsto \varnothing, \ldots, a_k \mapsto \varnothing, \textcolor{black}{c \mapsto \inct{}{u[\rho^n \mapsto q]}}, b_1 \mapsto t_1[\rho^{n+1} \mapsto \inct{}{q}], \ldots, b_n \mapsto t_n[\rho^{n+1} \mapsto \inct{}{q}]}\equiv \mkObject{a_1 \mapsto \varnothing, \ldots, a_k \mapsto \varnothing, \textcolor{black}{c \mapsto \inct{}{u}[\rho^{n+1} \mapsto \inct{}{q}]}, b_1 \mapsto t_1[\rho^{n+1} \mapsto \inct{}{q}], \ldots, b_n \mapsto t_n[\rho^{n+1} \mapsto \inct{}{q}]} \equiv \mkObject{a_1 \mapsto \varnothing, \ldots, a_k \mapsto \varnothing, \textcolor{black}{c \mapsto \inct{}{u}}, b_1 \mapsto t_1, \ldots, b_n \mapsto t_n}[\rho^n \mapsto q]$, with the lemma about swapping substitution and increment.
    \end{enumerate}
\end{proof}

\begin{lemma}[Substitution Lemma for $\innerparred$, Lemma~\ref{lemma:normal-reduction:substitution-inner}]
    If $t \innerparred s$ and $q \innerparred r$, then $t[\rho^n \mapsto q] \innerparred s[\rho^n \mapsto r]$.
\end{lemma}
\begin{proof}
By induction on $\innerparred$:
\begin{enumerate}
    \item cong$_{\text{OBJ}}^{\innerparred}$:
    If $\mkObject{a_1 \mapsto \varnothing, \ldots, a_k \mapsto \varnothing, b_1 \mapsto t_1, \ldots, b_n \mapsto t_n} \innerparred \mkObject{a_1 \mapsto \varnothing, \ldots, a_k \mapsto \varnothing, b_1 \mapsto t_1', \ldots, b_n \mapsto t_n'}$ and 
    $t_i \Rrightarrow t_i'$, then with the Substitution Lemma $t_i[\rho^{n+1} \mapsto q] \Rrightarrow t_i'[\rho^{n+1} \mapsto r]$, and consequently, $\mkObject{a_1 \mapsto \varnothing, \ldots, a_k \mapsto \varnothing, b_1 \mapsto t_1, \ldots, b_n \mapsto t_n}[\rho^n \mapsto q] \innerparred \mkObject{a_1 \mapsto \varnothing, \ldots, a_k \mapsto \varnothing, b_1 \mapsto t_1', \ldots, b_n \mapsto t_n'}[\rho^n \mapsto r]$.
    \item cong$_{\rho}^{\innerparred}$:
    If $\rho^i \innerparred \rho^i$, then 
    \begin{enumerate}
        \item if $i<n$, $\rho^i[\rho^n \mapsto q] \equiv \rho^i \innerparred \rho^i \equiv \rho^i[\rho^n \mapsto r]$,
        \item if $i=n$, $\rho^i[\rho^n \mapsto q] \equiv q \innerparred r \equiv \rho^i[\rho^n \mapsto r]$,
        \item if $i>n$, $\rho^i[\rho^n \mapsto q] \equiv \rho^{i-1} \innerparred \rho^{i-1} \equiv \rho^i[\rho^n \mapsto r]$,
    \end{enumerate}
    \item cong$_{\text{DOT}}^{\innerparred}$: 
    If $t.a \innerparred t'.a$ and $t \innerparred t'$, then by induction hypothesis $t[\rho^n \mapsto q] \innerparred t'[\rho^n \mapsto r]$, hence $t.a[\rho^n \mapsto q] \innerparred t'.a[\rho^n \mapsto r]$.
    \item cong$_{\text{APP}}^{\innerparred}$:
    If $t(a \mapsto u) \innerparred t'(a \mapsto u')$,
    and $u \Rrightarrow u'$ and $t \innerparred t'$, then by induction hypothesis $t[\rho^n \mapsto q] \innerparred t'[\rho^n \mapsto r]$, and by Substitution Lemma $u[\rho^n \mapsto q] \Rrightarrow u'[\rho^n \mapsto r]$. Hence $t(a \mapsto u)[\rho^n \mapsto q] \innerparred t'(a \mapsto u')[\rho^n \mapsto r]$.
\end{enumerate}
\end{proof}

\begin{lemma}[Standardizing Reductions, Lemma~\ref{lemma:normal-reduction:standardizing}]
    For any $\varphi$-terms $t, r, s$ such that $t \innerparred r \headred s$, there exists $\varphi$-term $q$, such that $t \headredmany q \innerparred s$.
\end{lemma}

\begin{proof}
By induction on the structure of $r \headred s$:
\begin{enumerate}
    \item $r\equiv p.a$ and $s \equiv p'.a$ with $p\headred p'$. Since $t \innerparred r$, $t \equiv p''.a$ with $p'' \innerparred p$. We have $p'' \innerparred p \headred p'$, and, by induction hypothesis, $p'' \headredmany q' \innerparred p'$. With congruence reduction rules, for $q\equiv q'.a$, $t \headredmany q \innerparred s$.
    \item $r\equiv p(a\mapsto u)$ and $s \equiv p'(a\mapsto u)$ with $p\headred p'$. Since $t \innerparred r$, $t \equiv p''(a\mapsto u')$ with $p'' \innerparred p$ and $u' \Rrightarrow u$. We have $p'' \innerparred p \headred p'$, and, by induction hypothesis, $p'' \headredmany q' \innerparred p'$. With congruence reduction rules, for $q\equiv q'(a\mapsto u')$, $t \headredmany q \innerparred s$.
    \item $r \equiv r'.c$ where $r' \equiv \mkObject{a_1 \mapsto \varnothing, \ldots, a_k \mapsto \varnothing, c \mapsto t_c, b_1 \mapsto t_1, \ldots, b_n \mapsto t_n}$, and $s \equiv t_c[\rho^0 \mapsto r']$. Since $t \innerparred r$, $t \equiv t'.c$ with $t' \innerparred r'$. Hence, $t' \equiv \mkObject{a_1 \mapsto \varnothing, \ldots, a_k \mapsto \varnothing, c \mapsto t_c', b_1 \mapsto t_1', \ldots, b_n \mapsto t_n'}$ with $t_c' \Rrightarrow t_c$ and $t_i' \Rrightarrow t_i$ for $i \in \{1, \ldots, n\}$. This implies that $t \headred t'_c[\rho^0 \mapsto t']$. By Substitution Lemma, $t'_c[\rho^0 \mapsto t'] \Rrightarrow t_c[\rho^0 \mapsto r']$, and by the Main Lemma, $t'_c[\rho^0 \mapsto t'] \headredmany p \innerparred t_c[\rho^0 \mapsto r']$ for some $p$. For $q\equiv p$, $t \headredmany q \innerparred s$.
    \item $r \equiv r'.c$ where $r' \equiv \mkObject{\ldots}$, $c \notin \attr(r')$, $\varphi \in \attr(r')$, and $s \equiv r'.\varphi.c$. Since $t \innerparred r$, $t \equiv t'.c$ with $t' \innerparred r'$, and $t' \equiv \mkObject{\ldots}$, $c \notin \attr(t')$, and $\varphi \in \attr(t')$. This implies that $t \headred t'.\varphi.c$. Since $t'\innerparred r'$, $t'.\varphi.c \innerparred r'.\varphi.c$. So, for $q\equiv t'.\varphi.c$, $t \headredmany q \innerparred s$.
    \item $r \equiv r'(c\mapsto u)$ where $r' \equiv  \mkObject{a_1 \mapsto \varnothing, \ldots, a_k \mapsto \varnothing, c \mapsto \varnothing, b_1 \mapsto t_1, \ldots, b_n \mapsto t_n}$, and $s \equiv \mkObject{a_1 \mapsto \varnothing, \ldots, a_k \mapsto \varnothing, c \mapsto \inct{}{u}, b_1 \mapsto t_1, \ldots, b_n \mapsto t_n}$. Since $t \innerparred r$, $t \equiv t'(c\mapsto u')$ with $t' \equiv \mkObject{a_1 \mapsto \varnothing, \ldots, a_k \mapsto \varnothing, c \mapsto \varnothing, b_1 \mapsto t_1', \ldots, b_n \mapsto t_n'}$, where $t_i' \Rrightarrow t_i$, and $u' \Rrightarrow u$. This implies that $t \headred \mkObject{a_1 \mapsto \varnothing, \ldots, a_k \mapsto \varnothing, c \mapsto \inct{}{u'}, b_1 \mapsto t_1', \ldots, b_n \mapsto t_n'} \equiv q$. So, $t \headredmany q \innerparred s$.
\end{enumerate}
\end{proof}

%% file: ms.bbl

\begin{thebibliography}{31}


\ifx \showCODEN    \undefined \def \showCODEN     #1{\unskip}     \fi
\ifx \showDOI      \undefined \def \showDOI       #1{#1}\fi
\ifx \showISBNx    \undefined \def \showISBNx     #1{\unskip}     \fi
\ifx \showISBNxiii \undefined \def \showISBNxiii  #1{\unskip}     \fi
\ifx \showISSN     \undefined \def \showISSN      #1{\unskip}     \fi
\ifx \showLCCN     \undefined \def \showLCCN      #1{\unskip}     \fi
\ifx \shownote     \undefined \def \shownote      #1{#1}          \fi
\ifx \showarticletitle \undefined \def \showarticletitle #1{#1}   \fi
\ifx \showURL      \undefined \def \showURL       {\relax}        \fi
\providecommand\bibfield[2]{#2}
\providecommand\bibinfo[2]{#2}
\providecommand\natexlab[1]{#1}
\providecommand\showeprint[2][]{arXiv:#2}

\bibitem[Abadi(1994)]%
        {Abadi1994}
\bibfield{author}{\bibinfo{person}{Martin Abadi}.}
  \bibinfo{year}{1994}\natexlab{}.
\newblock \showarticletitle{Baby Modula-3 and a theory of objects}.
\newblock \bibinfo{journal}{\emph{Journal of Functional Programming}}
  \bibinfo{volume}{4}, \bibinfo{number}{2} (\bibinfo{year}{1994}),
  \bibinfo{pages}{249–283}.
\newblock
\urldef\tempurl%
\url{https://doi.org/10.1017/S0956796800001052}
\showDOI{\tempurl}


\bibitem[Abadi and Cardelli(1996)]%
        {ABADI199678}
\bibfield{author}{\bibinfo{person}{Martín Abadi} {and} \bibinfo{person}{Luca
  Cardelli}.} \bibinfo{year}{1996}\natexlab{}.
\newblock \showarticletitle{A Theory of Primitive Objects: Untyped and
  First-Order Systems}.
\newblock \bibinfo{journal}{\emph{Information and Computation}}
  \bibinfo{volume}{125}, \bibinfo{number}{2} (\bibinfo{year}{1996}),
  \bibinfo{pages}{78--102}.
\newblock
\showISSN{0890-5401}
\urldef\tempurl%
\url{https://doi.org/10.1006/inco.1996.0024}
\showDOI{\tempurl}


\bibitem[Allen et~al\mbox{.}(2007)]%
        {AllenLuchangcoRyuSteele2007}
\bibfield{author}{\bibinfo{person}{Eric Allen}, \bibinfo{person}{J.~J.
  Hallett}, \bibinfo{person}{Victor Luchangco}, \bibinfo{person}{Sukyoung Ryu},
  {and} \bibinfo{person}{Guy~L. Steele}.} \bibinfo{year}{2007}\natexlab{}.
\newblock \showarticletitle{Modular Multiple Dispatch with Multiple
  Inheritance}. In \bibinfo{booktitle}{\emph{Proceedings of the 2007 ACM
  Symposium on Applied Computing}} (Seoul, Korea) \emph{(\bibinfo{series}{SAC
  '07})}. \bibinfo{publisher}{Association for Computing Machinery},
  \bibinfo{address}{New York, NY, USA}, \bibinfo{pages}{1117–1121}.
\newblock
\showISBNx{1595934804}
\urldef\tempurl%
\url{https://doi.org/10.1145/1244002.1244245}
\showDOI{\tempurl}


\bibitem[Allen et~al\mbox{.}(2011)]%
        {AllenHilburnKilpatrickLuchangcoRyuChaseSteele2011}
\bibfield{author}{\bibinfo{person}{Eric Allen}, \bibinfo{person}{Justin
  Hilburn}, \bibinfo{person}{Scott Kilpatrick}, \bibinfo{person}{Victor
  Luchangco}, \bibinfo{person}{Sukyoung Ryu}, \bibinfo{person}{David Chase},
  {and} \bibinfo{person}{Guy Steele}.} \bibinfo{year}{2011}\natexlab{}.
\newblock \showarticletitle{Type Checking Modular Multiple Dispatch with
  Parametric Polymorphism and Multiple Inheritance}.
\newblock \bibinfo{journal}{\emph{SIGPLAN Not.}} \bibinfo{volume}{46},
  \bibinfo{number}{10} (\bibinfo{date}{oct} \bibinfo{year}{2011}),
  \bibinfo{pages}{973–992}.
\newblock
\showISSN{0362-1340}
\urldef\tempurl%
\url{https://doi.org/10.1145/2076021.2048140}
\showDOI{\tempurl}


\bibitem[Barendregt and Barendsen(1984)]%
        {Barendregt1984}
\bibfield{author}{\bibinfo{person}{Henk Barendregt} {and} \bibinfo{person}{Erik
  Barendsen}.} \bibinfo{year}{1984}\natexlab{}.
\newblock \showarticletitle{Introduction to lambda calculus}.
\newblock \bibinfo{journal}{\emph{Nieuw archief voor wisenkunde}}
  \bibinfo{volume}{4} (\bibinfo{date}{01} \bibinfo{year}{1984}),
  \bibinfo{pages}{337--372}.
\newblock


\bibitem[Bugayenko(2021a)]%
        {bugayenko/online}
\bibfield{author}{\bibinfo{person}{Yegor Bugayenko}.}
  \bibinfo{year}{2021}\natexlab{a}.
\newblock \showarticletitle{EOLANG and $\varphi$-calculus}.
\newblock \bibinfo{journal}{\emph{arXiv preprint arXiv:2111.13384}}
  (\bibinfo{year}{2021}).
\newblock


\bibitem[Bugayenko(2021b)]%
        {bugayenko2021reducing}
\bibfield{author}{\bibinfo{person}{Yegor Bugayenko}.}
  \bibinfo{year}{2021}\natexlab{b}.
\newblock \showarticletitle{Reducing Programs to Objects}.
\newblock \bibinfo{journal}{\emph{arXiv preprint arXiv:2112.11988}}
  (\bibinfo{year}{2021}).
\newblock


\bibitem[Cardelli(1994)]%
        {Cardelli1994ExtensibleRI}
\bibfield{author}{\bibinfo{person}{Luca Cardelli}.}
  \bibinfo{year}{1994}\natexlab{}.
\newblock \showarticletitle{Extensible records in a pure calculus of
  subtyping}.
\newblock


\bibitem[Cardelli(1995)]%
        {Cardelli1995}
\bibfield{author}{\bibinfo{person}{Luca Cardelli}.}
  \bibinfo{year}{1995}\natexlab{}.
\newblock \showarticletitle{A Language with Distributed Scope}.
\newblock \bibinfo{journal}{\emph{Computing Systems}} \bibinfo{volume}{8},
  \bibinfo{number}{1} (\bibinfo{year}{1995}), \bibinfo{pages}{27--59}.
\newblock
\urldef\tempurl%
\url{http://www.usenix.org/publications/compsystems/1995/win\_cardelli.pdf}
\showURL{%
\tempurl}


\bibitem[Castagna et~al\mbox{.}(1992)]%
        {Castagna1992ACF}
\bibfield{author}{\bibinfo{person}{Giuseppe Castagna}, \bibinfo{person}{Giorgio
  Ghelli}, {and} \bibinfo{person}{Giuseppe Longo}.}
  \bibinfo{year}{1992}\natexlab{}.
\newblock \showarticletitle{A calculus for overloaded functions with
  subtyping}. In \bibinfo{booktitle}{\emph{LFP '92}}.
\newblock


\bibitem[Castegren and Wrigstad(2019)]%
        {CastegrenWrigstad2019}
\bibfield{author}{\bibinfo{person}{Elias Castegren} {and}
  \bibinfo{person}{Tobias Wrigstad}.} \bibinfo{year}{2019}\natexlab{}.
\newblock \showarticletitle{OOlong: A Concurrent Object Calculus for
  Extensibility and Reuse}.
\newblock \bibinfo{journal}{\emph{SIGAPP Appl. Comput. Rev.}}
  \bibinfo{volume}{18}, \bibinfo{number}{4} (\bibinfo{date}{jan}
  \bibinfo{year}{2019}), \bibinfo{pages}{47–60}.
\newblock
\showISSN{1559-6915}
\urldef\tempurl%
\url{https://doi.org/10.1145/3307624.3307629}
\showDOI{\tempurl}


\bibitem[Chlipala(2010)]%
        {Chlipala2010UrSM}
\bibfield{author}{\bibinfo{person}{Adam Chlipala}.}
  \bibinfo{year}{2010}\natexlab{}.
\newblock \showarticletitle{Ur: statically-typed metaprogramming with
  type-level record computation}. In \bibinfo{booktitle}{\emph{PLDI '10}}.
\newblock


\bibitem[Ciaffaglione et~al\mbox{.}(2021a)]%
        {Ciaffaglione21}
\bibfield{author}{\bibinfo{person}{Alberto Ciaffaglione},
  \bibinfo{person}{Pietro Di~Gianantonio}, \bibinfo{person}{Furio Honsell},
  {and} \bibinfo{person}{Luigi Liquori}.} \bibinfo{year}{2021}\natexlab{a}.
\newblock \showarticletitle{A prototype-based approach to object evolution.}
\newblock \bibinfo{journal}{\emph{The Journal of Object Technology}}
  \bibinfo{volume}{20} (\bibinfo{date}{01} \bibinfo{year}{2021}),
  \bibinfo{pages}{4:1}.
\newblock
\urldef\tempurl%
\url{https://doi.org/10.5381/jot.2021.20.2.a4}
\showDOI{\tempurl}


\bibitem[Ciaffaglione et~al\mbox{.}(2021b)]%
        {Ciaffaglione2021}
\bibfield{author}{\bibinfo{person}{Alberto Ciaffaglione},
  \bibinfo{person}{Pietro~Di Gianantonio}, \bibinfo{person}{Furio Honsell},
  {and} \bibinfo{person}{Luigi Liquori}.} \bibinfo{year}{2021}\natexlab{b}.
\newblock \showarticletitle{A prototype-based approach to object evolution}.
\newblock \bibinfo{journal}{\emph{The Journal of Object Technology}}
  \bibinfo{volume}{20}, \bibinfo{number}{2} (\bibinfo{year}{2021}),
  \bibinfo{pages}{4:1--24}.
\newblock
\urldef\tempurl%
\url{https://doi.org/10.5381/jot.2021.20.2.a4}
\showDOI{\tempurl}


\bibitem[de~Bruijn(1972)]%
        {deBruijn1972}
\bibfield{author}{\bibinfo{person}{Nicolaas~G. de Bruijn}.}
  \bibinfo{year}{1972}\natexlab{}.
\newblock \showarticletitle{Lambda calculus notation with nameless dummies, a
  tool for automatic formula manipulation, with application to the
  Church-Rosser theorem}.
\newblock \bibinfo{journal}{\emph{Indagationes Mathematicae (Proceedings)}}
  \bibinfo{volume}{75}, \bibinfo{number}{5} (\bibinfo{year}{1972}),
  \bibinfo{pages}{381--392}.
\newblock
\showISSN{1385-7258}
\urldef\tempurl%
\url{https://doi.org/10.1016/1385-7258(72)90034-0}
\showDOI{\tempurl}


\bibitem[Duke et~al\mbox{.}(1995)]%
        {DUKE1995511}
\bibfield{author}{\bibinfo{person}{Roger Duke}, \bibinfo{person}{Gordon Rose},
  {and} \bibinfo{person}{Graeme Smith}.} \bibinfo{year}{1995}\natexlab{}.
\newblock \showarticletitle{Object-Z: A specification language advocated for
  the description of standards}.
\newblock \bibinfo{journal}{\emph{Computer Standards \& Interfaces}}
  \bibinfo{volume}{17}, \bibinfo{number}{5} (\bibinfo{year}{1995}),
  \bibinfo{pages}{511--533}.
\newblock
\showISSN{0920-5489}
\urldef\tempurl%
\url{https://doi.org/10.1016/0920-5489(95)00024-O}
\showDOI{\tempurl}
\newblock
\shownote{Formal Description Techniques}.


\bibitem[Durr and Van~Katwijk(1992)]%
        {durr_vdm_1992}
\bibfield{author}{\bibinfo{person}{Eugene Durr} {and} \bibinfo{person}{Jan
  Van~Katwijk}.} \bibinfo{year}{1992}\natexlab{}.
\newblock \showarticletitle{{VDM}++, a formal specification language for
  object-oriented designs}. In \bibinfo{booktitle}{\emph{{CompEuro} 1992
  {Proceedings} {Computer} {Systems} and {Software} {Engineering}}}.
  \bibinfo{publisher}{IEEE Comput. Soc. Press}, \bibinfo{address}{The Hague,
  Netherlands}, \bibinfo{pages}{214--219}.
\newblock
\showISBNx{978-0-8186-2760-6}
\urldef\tempurl%
\url{https://doi.org/10.1109/CMPEUR.1992.218511}
\showDOI{\tempurl}


\bibitem[Fisher et~al\mbox{.}(1993)]%
        {Fisher1993ALC}
\bibfield{author}{\bibinfo{person}{Kathleen Fisher}, \bibinfo{person}{Furio
  Honsell}, {and} \bibinfo{person}{John~C. Mitchell}.}
  \bibinfo{year}{1993}\natexlab{}.
\newblock \showarticletitle{A lambda calculus of objects and method
  specialization}.
\newblock \bibinfo{journal}{\emph{Proceedings Eighth Annual IEEE Symposium on
  Logic in Computer Science}} (\bibinfo{year}{1993}), \bibinfo{pages}{26--38}.
\newblock


\bibitem[Gamma et~al\mbox{.}(1995)]%
        {GoF1995}
\bibfield{author}{\bibinfo{person}{Erich Gamma}, \bibinfo{person}{Richard
  Helm}, \bibinfo{person}{Ralph Johnson}, {and} \bibinfo{person}{John
  Vlissides}.} \bibinfo{year}{1995}\natexlab{}.
\newblock \bibinfo{booktitle}{\emph{Design Patterns: Elements of Reusable
  Object-Oriented Software}}.
\newblock \bibinfo{publisher}{Addison-Wesley Longman Publishing Co., Inc.},
  \bibinfo{address}{USA}.
\newblock
\showISBNx{0201633612}


\bibitem[Girard et~al\mbox{.}(1989)]%
        {Girard1989}
\bibfield{author}{\bibinfo{person}{Jean-Yves Girard}, \bibinfo{person}{Paul
  Taylor}, {and} \bibinfo{person}{Yves Lafont}.}
  \bibinfo{year}{1989}\natexlab{}.
\newblock \bibinfo{booktitle}{\emph{Proofs and Types}}.
\newblock \bibinfo{publisher}{Cambridge University Press},
  \bibinfo{address}{USA}.
\newblock
\showISBNx{0521371813}


\bibitem[Igarashi et~al\mbox{.}(2001)]%
        {IgarashiPierceWadler2001}
\bibfield{author}{\bibinfo{person}{Atsushi Igarashi},
  \bibinfo{person}{Benjamin~C. Pierce}, {and} \bibinfo{person}{Philip Wadler}.}
  \bibinfo{year}{2001}\natexlab{}.
\newblock \showarticletitle{Featherweight Java: A Minimal Core Calculus for
  Java and GJ}.
\newblock \bibinfo{journal}{\emph{ACM Trans. Program. Lang. Syst.}}
  \bibinfo{volume}{23}, \bibinfo{number}{3} (\bibinfo{date}{may}
  \bibinfo{year}{2001}), \bibinfo{pages}{396–450}.
\newblock
\showISSN{0164-0925}
\urldef\tempurl%
\url{https://doi.org/10.1145/503502.503505}
\showDOI{\tempurl}


\bibitem[Krivine(1993)]%
        {Krivine1993}
\bibfield{author}{\bibinfo{person}{Jean~L. Krivine}.}
  \bibinfo{year}{1993}\natexlab{}.
\newblock \bibinfo{booktitle}{\emph{Lambda-Calculus, Types and Models}}.
\newblock \bibinfo{publisher}{Ellis Horwood}, \bibinfo{address}{USA}.
\newblock
\showISBNx{0130624071}


\bibitem[\"{O}stlund and Wrigstad(2010)]%
        {OstlundWrigstad2010}
\bibfield{author}{\bibinfo{person}{Johan \"{O}stlund} {and}
  \bibinfo{person}{Tobias Wrigstad}.} \bibinfo{year}{2010}\natexlab{}.
\newblock \showarticletitle{Welterweight Java}. In
  \bibinfo{booktitle}{\emph{Proceedings of the 48th International Conference on
  Objects, Models, Components, Patterns}} (M\'{a}laga, Spain)
  \emph{(\bibinfo{series}{TOOLS'10})}. \bibinfo{publisher}{Springer-Verlag},
  \bibinfo{address}{Berlin, Heidelberg}, \bibinfo{pages}{97–116}.
\newblock
\showISBNx{3642139523}


\bibitem[Park et~al\mbox{.}(2019)]%
        {ParkHongSteeleRyu2019}
\bibfield{author}{\bibinfo{person}{Gyunghee Park}, \bibinfo{person}{Jaemin
  Hong}, \bibinfo{person}{Guy~L. Steele~Jr.}, {and} \bibinfo{person}{Sukyoung
  Ryu}.} \bibinfo{year}{2019}\natexlab{}.
\newblock \showarticletitle{Polymorphic Symmetric Multiple Dispatch with
  Variance}.
\newblock \bibinfo{journal}{\emph{Proc. ACM Program. Lang.}}
  \bibinfo{volume}{3}, \bibinfo{number}{POPL}, Article \bibinfo{articleno}{11}
  (\bibinfo{date}{jan} \bibinfo{year}{2019}), \bibinfo{numpages}{28}~pages.
\newblock
\urldef\tempurl%
\url{https://doi.org/10.1145/3290324}
\showDOI{\tempurl}


\bibitem[Pierce and Turner(1993)]%
        {pierce_simple_1993}
\bibfield{author}{\bibinfo{person}{Benjamin~C. Pierce} {and}
  \bibinfo{person}{David~N. Turner}.} \bibinfo{year}{1993}\natexlab{}.
\newblock \bibinfo{title}{Simple {Type}-{Theoretic} {Foundations} for
  {Object}-{Oriented} {Programming}}.
\newblock
\newblock


\bibitem[Salzman and Aldrich(2005)]%
        {SalzmanAldrich2005}
\bibfield{author}{\bibinfo{person}{Lee Salzman} {and} \bibinfo{person}{Jonathan
  Aldrich}.} \bibinfo{year}{2005}\natexlab{}.
\newblock \showarticletitle{Prototypes with Multiple Dispatch: An Expressive
  and Dynamic Object Model}. In \bibinfo{booktitle}{\emph{Proceedings of the
  19th European Conference on Object-Oriented Programming}} (Glasgow, UK)
  \emph{(\bibinfo{series}{ECOOP'05})}. \bibinfo{publisher}{Springer-Verlag},
  \bibinfo{address}{Berlin, Heidelberg}, \bibinfo{pages}{312–336}.
\newblock
\showISBNx{354027992X}
\urldef\tempurl%
\url{https://doi.org/10.1007/11531142_14}
\showDOI{\tempurl}


\bibitem[Takahashi(1995)]%
        {Takahashi1995}
\bibfield{author}{\bibinfo{person}{Masako Takahashi}.}
  \bibinfo{year}{1995}\natexlab{}.
\newblock \showarticletitle{Parallel Reductions in $\lambda$-Calculus}.
\newblock \bibinfo{journal}{\emph{Information and Computation}}
  \bibinfo{volume}{118}, \bibinfo{number}{1} (\bibinfo{year}{1995}),
  \bibinfo{pages}{120--127}.
\newblock
\showISSN{0890-5401}
\urldef\tempurl%
\url{https://doi.org/10.1006/inco.1995.1057}
\showDOI{\tempurl}


\bibitem[Ungar and Smith(1987)]%
        {UngarSmith1987}
\bibfield{author}{\bibinfo{person}{David~M. Ungar} {and}
  \bibinfo{person}{Randall~B. Smith}.} \bibinfo{year}{1987}\natexlab{}.
\newblock \showarticletitle{SELF: The power of simplicity}.
\newblock \bibinfo{journal}{\emph{LISP and Symbolic Computation}}
  \bibinfo{volume}{4} (\bibinfo{year}{1987}), \bibinfo{pages}{187--205}.
\newblock


\bibitem[Wand(1987)]%
        {Wand1987CompleteTI}
\bibfield{author}{\bibinfo{person}{Mitchell Wand}.}
  \bibinfo{year}{1987}\natexlab{}.
\newblock \showarticletitle{Complete Type Inference for Simple Objects}. In
  \bibinfo{booktitle}{\emph{LICS}}.
\newblock


\bibitem[Wand(1991)]%
        {WAND19911}
\bibfield{author}{\bibinfo{person}{Mitchell Wand}.}
  \bibinfo{year}{1991}\natexlab{}.
\newblock \showarticletitle{Type inference for record concatenation and
  multiple inheritance}.
\newblock \bibinfo{journal}{\emph{Information and Computation}}
  \bibinfo{volume}{93}, \bibinfo{number}{1} (\bibinfo{year}{1991}),
  \bibinfo{pages}{1--15}.
\newblock
\showISSN{0890-5401}
\urldef\tempurl%
\url{https://doi.org/10.1016/0890-5401(91)90050-C}
\showDOI{\tempurl}
\newblock
\shownote{Selections from 1989 IEEE Symposium on Logic in Computer Science}.


\bibitem[Wang et~al\mbox{.}(2018)]%
        {WangZhangOliveiraServetto2018}
\bibfield{author}{\bibinfo{person}{Yanlin Wang}, \bibinfo{person}{Haoyuan
  Zhang}, \bibinfo{person}{Bruno~C. d. S.~Oliveira}, {and}
  \bibinfo{person}{Marco Servetto}.} \bibinfo{year}{2018}\natexlab{}.
\newblock \showarticletitle{{FHJ:} {A} Formal Model for Hierarchical
  Dispatching and Overriding}. In \bibinfo{booktitle}{\emph{32nd European
  Conference on Object-Oriented Programming, {ECOOP} 2018, July 16-21, 2018,
  Amsterdam, The Netherlands}} \emph{(\bibinfo{series}{LIPIcs},
  Vol.~\bibinfo{volume}{109})}, \bibfield{editor}{\bibinfo{person}{Todd~D.
  Millstein}} (Ed.). \bibinfo{publisher}{Schloss Dagstuhl - Leibniz-Zentrum
  f{\"{u}}r Informatik}, \bibinfo{pages}{20:1--20:30}.
\newblock
\urldef\tempurl%
\url{https://doi.org/10.4230/LIPIcs.ECOOP.2018.20}
\showDOI{\tempurl}


\end{thebibliography}
